\documentclass[11pt,letterpaper,english,letter, 11pt]{article}
\usepackage[T1]{fontenc}
\usepackage[latin9]{inputenc}
\usepackage[active]{srcltx}
\usepackage{color}
\usepackage{float}
\usepackage{calc}
\usepackage{amsmath}
\usepackage{amssymb}
\usepackage{graphicx}
\usepackage{setspace}
\usepackage{esint}

\makeatletter

\usepackage{graphics}\usepackage{amsthm}\usepackage{amsopn}\usepackage{amstext}\usepackage{amsfonts}\@ifundefined{definecolor}
 {\usepackage{color}}{}
\usepackage{fullpage}\usepackage{hyperref}

\newcommand{\beq}{\begin{equation}}
\newcommand{\eeq}{\end{equation}}

\newcommand{\epc}{\hspace{1pc}}

\newcommand{\wh}{\widehat}

\newtheorem{theorem}{Theorem}
\newtheorem{proposition}[theorem]{Proposition}
\newtheorem{lemma}[theorem]{Lemma}

\newtheorem{definition}[theorem]{Definition}

\newtheorem{remark}[theorem]{Remark}

\makeatother

\usepackage{babel}
\begin{document}

\title{{\Huge {\vspace{-2cm}}}\textbf{\Huge Joint Sensing and Power Allocation
in Nonconvex Cognitive Radio Games: Quasi-Nash Equilibria{\vspace{-0.3cm}}}}

\author{{\normalsize Jong-Shi Pang$^{1}$ and Gesualdo Scutari$^{2}$}{\small{}
}\\
{\small{} $^{1}$Dept. of Industrial and Enterprise Systems Engineering,
University of Illinois, Urbana, IL 61801, U.S.A. }\\
{\small $^{2}$Dept. of Electrical Engineering, State University of
New York at Buffalo, Buffalo, NY 14260, U.S.A. }\\
{\small{} Emails:}\texttt{\small jspang@illinois.edu }{\small and}\texttt{\small{}
gesualdo@buffalo.edu.}{\normalsize }%
\thanks{The work of Pang is based on research supported by the U.S.A. National
Science Foundation grant CMMI 0969600 and by the Air Force Office
of Sponsored Research award No. FA9550-09-10329. The work of Scutari
was supported by U.S.A. National Science Foundation grant CSM 1218717.
{\small }\protect \\
\indent Copyright (c) 2012 IEEE. Personal use of this material is
permitted. However, permission to use this material for any other
purposes must be obtained from the IEEE by sending a request to pubs-permissions@ieee.org.%
}{\normalsize \vspace{-.2cm}}}

\date{{\normalsize Submitted to }\emph{\normalsize IEEE Transactions on}{\normalsize{}
}\emph{\normalsize Signal Processing}{\normalsize , April 7, 2012.
Second revision, Dec. 25, 2012. \vspace{-.9cm}}}
\maketitle
\begin{abstract}
\noindent In this paper, we propose a novel class of Nash problems
for Cognitive Radio (CR) networks composed of multiple primary users
(PUs) and secondary users (SUs) wherein each SU (player) competes
against the others to maximize his own opportunistic throughput by
choosing \emph{jointly} the sensing duration, the detection thresholds,
\emph{and} the vector power allocation over a multichannel link. In
addition to power budget constraints, several (deterministic or probabilistic)
interference constraints can be accommodated in the proposed general
formulation, such as constraints on the maximum individual/aggregate
(probabilistic) interference tolerable from the PUs. To keep the optimization
as decentralized as possible, global interference constraints, when
present, are imposed via pricing; the prices are thus additional variables
to be optimized. The resulting players' optimization problems are
nonconvex and there are price clearance conditions associated with
the nonconvex global interference constraints to be satisfied by the
equilibria of the game, which make the analysis of the proposed game
a challenging task; none of classical results in the game theory literature
can be successfully applied. To deal with the nonconvexity of the
game, we introduce a relaxed equilibrium concept$-$the Quasi-Nash
Equilibrium (QNE)$-$and study its main properties, performance, and
connection with local Nash equilibria. Quite interestingly, the proposed
game theoretical formulations yield a considerable performance improvement
with respect to current centralized and decentralized designs of CR
systems, 
which validates the concept of QNE.{\normalsize \vspace{-0.3cm}} 
\end{abstract}

\section{Introduction and Motivation}

In the last decade, Cognitive Radio (CR) has received considerable
attention as a way to improve the efficiency of radio networks \cite{Mitola_MoMuC99,Haykin_JSAC05}.
The CR networks adopt a hierarchical access structure where the primary
users (PUs) are the legacy spectrum holders while the secondary users
(SUs) are the unlicensed users who sense the electromagnetic environment
and adapt their transceiver parameters as well as the resource allocation
decisions in order to dynamically access temporally unoccupied spectrum
regions. 

The challenge for a reliable sensing algorithm is to identify suitable
transmission opportunities without compromising the integrity of the
PUs. One of the design criteria is to make the probability of false
alarm as low as possible, since it measures the percentage of vacant
spectrum which is misclassified as busy, increasing thus the opportunistic
usage of the spectrum from the SUs. On the other hand, in order to
limit the probability of interfering with PUs, it is desirable to
keep the missed detection probability as low as possible. The detection
thresholds are the trade-off factor between the false alarm and the
missed detection probabilities; low thresholds will result in high
false alarm rate in favor of low missed detection probability and
vice versa. Alternatively, the choice of the sensing time offers a
trade-off between the quality and speed of sensing: increasing the
sensing times permits to reach both low false alarm and detection
probability values, thus reducing the time available for secondary
transmissions, which would result in low SUs' throughput. The above
trade-off calls naturally for a joint determination of the sensing
and transmission parameters of the SUs, under the paradigm of selfish
behavior among these users. The modeling and analysis of this competitive
multi-agent optimization is the main subject of this work.

\subsection{Related work\vspace{-0.2cm}}

The joint optimization of the sensing and transmission parameters
has been only \emph{partially} addressed in the literature; 
current works on the subject can be divided in the following three
main classes.\emph{\smallskip{}
}

\noindent \textbf{-} \textbf{Optimization of the sensing parameters
only}\emph{ }\cite{Quan-Cui-Poor-Sayed,Quan-Cui-Poor-Sayed_TSP09,Liang-Zeng-Peh-Hoang_TWC08,Kim-Giannakis08,Rongfei-Hai_TWC10,Hoseini_Beaulieu_ICC10}\emph{.}
In \cite{Quan-Cui-Poor-Sayed,Quan-Cui-Poor-Sayed_TSP09}, the authors
proposed alternative centralized schemes that optimize the detection
thresholds for a bank of energy detectors, in order to maximize the
so-called\textit{\emph{ opportunistic throughput, while keeping the
sensing time and the transmission parameters of the SU fixed and given
a priori. }}The optimization of the sensing time and the sensing time/detection
thresholds for a given missed detection probability and transmission
rate of one SU was addressed in \cite{Liang-Zeng-Peh-Hoang_TWC08}
and \cite{Rongfei-Hai_TWC10,Hoseini_Beaulieu_ICC10}, respectively.
A throughput-sensing trade-off for a given transmission rate of the
SU was studied in \cite{Kim-Giannakis08}. In the above papers there
is no optimization of the SU's transmission strategies, and the proposed
schemes are applicable only to CR scenarios composed by \emph{one
}PU\emph{ }and SU\emph{.\smallskip{}
 }

\noindent \textbf{-} \textbf{Optimization of the transmission parameters
only}\emph{ }\cite{Xing-Chetan-Mathur-Haleem-Chandramouli-Subbalakshmi_TranMoBComp07_IntTempLim,Lu-WWang-TWang-Peng_GLOB08,WangPengWang_WCNC07,Luo-Pang_IWFA-Eurasip,Scutari-Palomar-Barbarossa_SP08_PI,Scutari-Palomar-Barbarossa_AIWFA_IT08,Leshem-Zehavip_SPMag_sub09,Scutari-Palomar-Facchinei-Pang_SPMag09,Pang-Scutari-Palomar-Facchinei_SP_10,SchmidtShiBerryHonigUtschick-SPMag}\emph{.}
These papers address the optimization of the SUs' transceivers in
a \emph{multiuser} OFDM CR scenario; several centralized \cite{Xing-Chetan-Mathur-Haleem-Chandramouli-Subbalakshmi_TranMoBComp07_IntTempLim,Lu-WWang-TWang-Peng_GLOB08,WangPengWang_WCNC07,SchmidtShiBerryHonigUtschick-SPMag}
schemes or distributed algorithms based on a game theoretic approach
\cite{Luo-Pang_IWFA-Eurasip,Scutari-Palomar-Barbarossa_SP08_PI,Scutari-Palomar-Barbarossa_AIWFA_IT08,Leshem-Zehavip_SPMag_sub09}
were proposed; a recent overview of the latter approaches can be found
in \cite{Leshem-Zehavip_SPMag_sub09,Larsson-Jorswieck-Lindblom-Mochaourab_SPMag_sub09}.
A general framework based on variational inequalities was proposed
in \cite{Scutari-Palomar-Facchinei-Pang_SPMag09,Pang-Scutari-Palomar-Facchinei_SP_10}
to study and solve in a distributed way convex noncooperative games
with (possibly) side (i.e., coupling) constraints (e.g., interference
constraints). In all the aforementioned papers the sensing process
is not considered as part of the optimization; in fact the SUs do
not perform any sensing but they are allowed to transmit over the
licensed spectrum provided that they satisfy interference constraints
imposed by the PUs, 
no matter if the PUs are active of not. \emph{\smallskip{}
}

\noindent \textbf{-} \textbf{Joint optimization of }\textbf{\emph{some}}\textbf{
of the sensing/transmission parameters}\emph{ }\cite{Kang-Liang-Garg-ZhangVT09,Pei-Liang-Teh-Li_TWC09,Stotas-Nallanathan_TC11,Barbarossa-Sardellitti-Scutari_CAMSAP09,Barbarossa-Sardellitti-Scutari_Cogis09,Huang-Lozano_ICASSP11}\emph{.}
In \cite{Kang-Liang-Garg-ZhangVT09} (or \cite{Pei-Liang-Teh-Li_TWC09,Stotas-Nallanathan_TC11}),
the sensing time and the transmit power (or the power allocation \cite{Pei-Liang-Teh-Li_TWC09,Stotas-Nallanathan_TC11}
over multi-channel links) of one SU were jointly optimized while keeping
the detection probability (and thus the decision threshold) fixed
to a target value. In \cite{Barbarossa-Sardellitti-Scutari_CAMSAP09,Barbarossa-Sardellitti-Scutari_Cogis09},
the authors focused on the joint optimization of the power allocation
and the equi-false alarm rate of one SU \cite{Barbarossa-Sardellitti-Scutari_CAMSAP09}
over multi-channel links, for a \emph{fixed} sensing time. The case
of multiple SUs and one PU was considered in \cite{Barbarossa-Sardellitti-Scutari_Cogis09}
(and more recently in \cite{Huang-Lozano_ICASSP11}), under the same
assumptions of \cite{Barbarossa-Sardellitti-Scutari_CAMSAP09}; however
no formal analysis of the proposed formulation was provided. Moreover,
these papers \cite{Barbarossa-Sardellitti-Scutari_CAMSAP09,Barbarossa-Sardellitti-Scutari_Cogis09,Huang-Lozano_ICASSP11}
did not consider the sensing overhead as part of the optimization,
leaving unclear how to optimally choose the sensing time. 

\vspace{-0.2cm}

\subsection{Main contributions \label{sub:Main-contributions}\vspace{-0.2cm}}

What emerges from the analysis of existing literature is that, up
to now, state-of-the-art approaches have focused on the optimization
of specific features and single components of a CR system in isolation.
In this paper we move a step ahead and propose a novel class of Nash
Equilibrium Problems (NEPs), suitable for designing multiple primary/secondary
user CR networks, under different practical settings. Aiming to dynamically
access temporally unoccupied primary spectrum regions, in our formulation,
the SUs maximize their own opportunistic throughput by \emph{jointly}
optimizing the sensing time, the decision thresholds of a bank of
energy detectors, and the power allocation over multi-channel links.
Because of sensing errors, it may happen however that the SUs attempt
to access part of licensed spectrum that is still occupied by active
PUs, causing thus harmful interference. This motivates the introduction
of \emph{probabilistic} (rather than deterministic) interference constraints
that are imposed to control the power radiated over the licensed spectrum
\emph{only when} a missed detection event occurs (in a probabilistic
sense). The proposed formulation accommodates alternative combinations
of power/interference constraints. For instance, in addition to classical
(deterministic) transmit power (and possibly spectral masks) constraints,
we envisage the use of (probabilistic) average individual (i.e., on
each SU) and/or global (i.e., over all the SUs) interference tolerable
by the primary receivers. Local interference constraints fit naturally
to scenarios where the SUs are not willing to cooperate; whereas the
global ones, which are less conservative, are more suitable for settings
where SUs may want to trade some signaling for better performance.
By imposing a coupling among the power allocations of the SUs, global
interference constraints introduce a new challenge in the system design:
how to enforce global interference constraints without requiring a
centralized optimization? A major contribution of the paper is to
address this issue by introducing a pricing mechanism in the game,
through a penalization in the players\textquoteright{} objective functions.
The prices will be then additional variables to be optimized, while
guaranteeing the global interference constraints to be satisfied at
any solution of the game.

The resulting game-theoretical formulations belong to the class of
nonconvex games, where the nonconvexity occurs at both the objective
functions and feasible sets of the SUs' optimization problems. Indeed
the objective functions of the SUs (to be maximized) are not quasi-concave
and the local/global interference constraints are bi-convex and thus
nonconvex; moreover, in the presence of global interference constraints,
there are possibly unbounded price variables to be optimized and pricing
clearing conditions associated with these constraints to be satisfied
at the equilibrium. All this makes the analysis of the proposed games
a challenging task; none of previous results in the game theory literature
can be successfully applied (see Sec.\! \ref{sec:NE-LNE-QE} for
more details); without (quasi-)concavity of the players' payoff functions
or convexity of the strategy sets, a solution of the game$-$the Nash
Equilibrium (NE)$-$may not even exist \cite{Baye-Tian-Zhou_RES93,Rosen_econ65}.
To overcome this issue, alternative refinements of the NE concept
have been introduced in the literature, in the form of ``local''
solutions (see, e.g., \cite{Baye-Tian-Zhou_RES93,vanDamme_book,Ferrer_AniaEcLetters01,CornetaCzarnecki01}
and references therein): in contrast to the NE that is resistant to
``arbitrarily large'' deviations of single player's strategies,
\emph{local} equilibria require stability only against ``small''
unilaterally deviations of the players. Among all, we mention here
some local solution concepts that have been proposed in the literature:
such as the ``critical NE'' \cite{Baye-Tian-Zhou_RES93}, the ``Local
NE'' (LNE) \cite{Ferrer_AniaEcLetters01}, and the ``generalized
equilibrium'' \cite{CornetaCzarnecki01}. Despite their theoretical
interest, those local solution concepts have a very limited applicability:
when available, there are only abstract mathematical conditions granting
their existence (see, e.g., \cite{Baye-Tian-Zhou_RES93,CornetaCzarnecki01}),
whose verification for games arising from realistic applications such
as those presented in the present paper remains elusive. A system
design based on these concepts may lead to an unpredictable system
behavior (e.g., equilibrium existence vs. nonexistence, convergence
of algorithms vs. nonconvergence), which is not a desirable feature
in practice. 

In this paper, we propose an alternative formalization of solution
in nonconvex games; building on the well-established and accepted
definition of stationary solutions of a nonlinear programming, we
introduce a (relaxed) equilibrium concept for a nonconvex NEP, which
is a solution of the aggregated stationarity conditions (the Karush-Kuhn-Tucker
conditions) of the players' optimization problems. We call such a
stationary tuple a \textsl{quasi-Nash equilibrium} (QNE) of the (nonconvex)
game; the prefix ``quasi'' is intended to signify that a NE (if
it exists) must be a QNE under mild conditions (constraint qualifications).
While the concept of QNE seems quite natural and conceptually similar
to the idea of critical NE \cite{Baye-Tian-Zhou_RES93} and generalized
equilibrium \cite{CornetaCzarnecki01}, its analysis is not an easy
task (cf. Sec.\!\!\! \ref{sec:NE-LNE-QE}), and cannot be addressed
using results in previous papers (e.g., \cite{Baye-Tian-Zhou_RES93,CornetaCzarnecki01});
for instance, in the absence of convexity or boundedness of the price
variables that could lead to the non-existence of a NE, the existence
of a QNE is by no means obvious. Building on our recent results in
\cite{PScutari10}, a major contribution of the paper is to provide
a satisfactory characterization of the QNE, in terms of existence,
connection with the (L)NE, and performance. The main features of the
QNE are the following: i) it \emph{always} exists for the proposed
class of nonconvex games; ii) it is shown to have some local optimality
properties (in the sense described in Sec.\! \ref{th:QE}); and iii)
the proposed joint optimization of sensing and transmission strategies
based on the QNE yields a considerable performance improvements (more
than one hundred per cent in ``high'' interference regimes) with
respect to both\emph{ centralized} \cite{SchmidtShiBerryHonigUtschick-SPMag}
and\emph{ decentralized} \cite{Luo-Pang_IWFA-Eurasip,Scutari-Palomar-Barbarossa_AIWFA_IT08,Leshem-Zehavip_SPMag_sub09,Pang-Scutari-Palomar-Facchinei_SP_10}
CR designs that do not perform any joint optimization. These desired
properties validate the use of the QNE both theoretically and practically. 

At the QNE of the game, the optimal sensing times of the SUs tend
to be different, implying that some SUs may start transmitting while
some others are still in the sensing phase, which may introduce a
significant performance degradation in the sensing process of the
other SUs. To overcome this issue, a further contribution of the paper
is to propose a distributed approach to ``force'' (at least asymptotically)
the same optimal sensing time for all the SUs; which is referred to
as\emph{ the equi-sensing case.} 

\textcolor{black}{In summary, the main contributions of the paper
are the following:}\\
\textcolor{black}{$\bullet$ We propose a novel class of NEPs (possibly
with}\textcolor{black}{\emph{ }}\textcolor{black}{pricing)}\textcolor{black}{\emph{,}}\textcolor{black}{{}
where for the first time a }\textcolor{black}{\emph{joint optimization}}\textcolor{black}{{}
of the detection threshold, the sensing time, and the transmission
strategies of multiple SUs is performed, under power and several (novel)
local and/or global probabilistic interference constraints;}\\
\textcolor{black}{$\bullet$ We introduce a relaxed equilibrium concept$-$the
QNE$-$and develop a novel optimization-based theory for studying
its main properties, the connection with the LNE, and its performance;}\\
\textcolor{black}{$\bullet$ We show how to modify the original NEPs
to impose in a distributed way an optimal equi-sensing time to all
the SUs, and study the main properties of the resulting game.}\\
\textcolor{black}{Overall, we introduce a new line of analysis based
on Variational Inequalities (VIs) in the literature of nonconvex games
(possibly) with pricing that is expected to be broadly applicable
for other game models.}

\textcolor{black}{Because of the space limitation, we omit here details
on how to solve the proposed games, which is the subject of a companion
paper \cite{Pang-Scutari-NNConvex_PII}, where we propose alternative
}\textcolor{black}{\emph{distributed}}\textcolor{black}{{} algorithms
along with their convergence properties and a detailed discussion
on their practical implementation; and also establish the existence
of NE of the games under a set of small-interference assumptions between
the SUs and PUs (that are }\textcolor{black}{\emph{not}}\textcolor{black}{{}
needed for the QNE). }

\textcolor{black}{The rest of the paper is organized as follows. Sec.}\!\!\textcolor{black}{{}
\ref{sec:System-Model} describes the system model, whereas Sec. \ref{sub:-Game-theoretical-formulation}
introduces the new game-theoretical formulations. The analysis of
the game is addressed in Sec.}\!\!\textcolor{black}{{} \ref{sec:NE-LNE-QE},
where the LNE and the concept of QNE are introduced, and their properties
and connection are studied. Sec.}\!\!\textcolor{black}{{} \ref{sec:The-Equi-sensing-Case}
focuses on the equi-sensing case; whereas Sec.}\!\!\textcolor{black}{{}
\ref{sec:Extensions-and-Generalizations} outlines some extensions
of the proposed formulations covering more general settings. Sec.}\!\!\textcolor{black}{{}
\ref{sec:Numerical-Results} provides some numerical results showing
the quality of the proposed approach;  and finally Sec.}\!\textcolor{black}{{}
\ref{sec:Conclusions} draws the conclusions.}\textcolor{red}{{} }\vspace{-.3cm}

\section{System Model \label{sec:System-Model}\vspace{-0.3cm}}

We consider a scenario composed of $Q$ active SUs, each formed by
a transmitter-receiver pair, coexisting in the same area and sharing
the same band. We focus on block transmissions over SISO frequency-selective
channels. It is well known that, in such a case, multicarrier transmission
is capacity achieving for large block-length. Based on the 802.22-Wireless-Regional-Area-Networks
standard (WRAN) \cite{FCC_WGWRAN}, we assume that the Medium Access
Control (MAC) frame is divided in two slots: one sensing slot and
one data slot. During the sensing interval, the SUs stay silent and
sense the electromagnetic environment looking for the spectrum holes,
whereas during the data slot they transmit (possibly) simultaneously
over the portions of the licensed spectrum detected as available.
The sensing and transmission phases are described in Sec. \ref{sub:Detection-problem}
and Sec. \ref{sub:The-transmission-phase}, respectively, whereas
the proposed joint optimization of the sensing and transmission strategies
over the two phases is introduced in Sec. \ref{sub:-Game-theoretical-formulation}. 

We wish to point out that, even if we focus on ad-hoc networks and
adopt the classical CR terminology (e.g., PUs, SUs, etc.), the proposed
model and results are applicable to broader scenarios, wherever there
are heterogeneous (small-cell) networks organized in a hierarchical
or multi-tier structure and sharing the same spectrum (e.g., due to
the spatial reuse of the frequencies). This happens for example in
femtocell networks, where femto access points (FAPs) serve femto users,
subject to some constraints on the interference radiated toward mobile
users served by macro base stations. \vspace{-0.2cm}

\subsection{The spectrum sensing problem\label{sub:Detection-problem}\vspace{-0.2cm}}

For notational simplicity, we preliminarily introduce the sensing
problem in a CR scenario composed of one active PU and under some
simplified assumptions; a more general setting (e.g., multiple active
PUs) is addressed in Sec. \ref{sec:Extensions-and-Generalizations}
within the framework of robust detection. 

The spectrum sensing problem of SU $q=1,\ldots,Q$ on subcarrier $k=1,\ldots,N$
is formulated as a binary hypothesis testing, with the following two
sets of hypotheses: at time index $n=1,2,\ldots,K_{q},$\emph{\vspace{-.2cm}}
\begin{equation}
\begin{array}{l}
\mbox{(PU signal absent) }\qquad\,\,\mathcal{H}_{k|0}:\quad y_{q,k}[n]=w_{q,k}[n]\\
\mbox{(PU signal present) }\qquad\mathcal{H}_{k|1}:\quad y_{q,k}[n]=I_{q,k}[n]+w_{q,k}[n],\vspace{-.2cm}
\end{array}\label{eq:Hyp_testing_composite}
\end{equation}
where $\mathcal{H}_{k|0}$ represents the absence of any primary signal
over the subcarrier $k$$-$the received baseband complex signal $y_{q,k}[n]$
contains only additive background noise $w_{q,k}[n]$$-$and $\mathcal{H}_{k|1}$
represents the presence of (at least) one PU$-$the received signal
contains the primary signaling $I_{q,k}[n]$ corrupted by noise$-$and
$K_{q}=\left\lfloor \tau_{q}\, f_{q}\right\rfloor \simeq\tau_{q}f_{q}$
is the number of samples, with $\tau_{q}$ and $f_{q}$ denoting the
sensing time and the sampling frequency, respectively.

Given the hypothesis testing problem in (\ref{eq:Hyp_testing_composite}),
we make the following standard assumptions in the literature of sensing
algorithms \cite{Quan-Cui-Poor-Sayed,Quan-Cui-Poor-Sayed_TSP09,Liang-Zeng-Peh-Hoang_TWC08,Kim-Giannakis08,Rongfei-Hai_TWC10,Hoseini_Beaulieu_ICC10,Kang-Liang-Garg-ZhangVT09,Pei-Liang-Teh-Li_TWC09,Barbarossa-Sardellitti-Scutari_CAMSAP09}:
$(A1)$ the noise $w_{q,k}[n]$ is drawn by a circularly symmetric
complex stationary (ergodic) Gaussian process with zero mean and variance
$\sigma_{q,k}^{2}$; $(A2)$ the primary signaling $I_{q,k}[n]$ are
samples of a circularly symmetric complex stationary (ergodic) random
process with zero mean, variance $\sigma_{I_{q,k}}^{\text{2}}$, and
statistically independent of the noise process; $(A3)$ the sample
frequency $f_{q}$ is chosen so that the random variables (RVs) $\{y_{q,k}[n]\}_{n}$
are independent, for all $q$ and $k$;%
\footnote{We can relax the constraint on $f_{q}$ by requiring in (A3) instead
block independence of the RVs $\{y_{q,k}[n]\}_{n}$, and invoke for
our derivations the central limit theorem valid for correlated ($m$-indepedend)
RVs, as stated, e.g., in \cite[Th. 2.8.1]{Lehmann_Devbook98}; because
of the space limitation, we omit these details here and stay within
the original assumption (A3). %
} $(A4)$ the system parameters$-$the spectrum occupancy status of
the PUs and the primary/secondary (cross-)channels$-$change sufficiently
slowly such that they can be assumed to be constant over each sensing/transmission
interval; and $(A5)$ the noise variance $\sigma_{q,k}^{2}$ and the
power $\sigma_{I_{q,k}}^{\text{2}}$ of the primary signaling are
estimate with no errors at the secondary receiver $q$. Note that
$(A5)$ can be relaxed by taking explicitly into account estimation
errors in the knowledge of $\sigma_{q,k}^{2}$ and $\sigma_{I_{q,k}}^{\text{2}}$
as well as the effect of shadowing and fading, as outlined in Sec.
\ref{sec:Extensions-and-Generalizations}; see also \cite{Scutari-Pang_DSP11}. 

Within the Neyman-Pearson framework, the test statistic of SU $q$
over subcarrier $k$ based on the energy detector maximizing the detection
probability under a given false alarm rate is \cite{Levy-book-Det-Est08}:\emph{\vspace{-.2cm}}
\begin{equation}
D\left(\mathbf{Y}_{q,k}\right)\triangleq\dfrac{{1}}{K_{q}}\sum_{n=1}^{K_{q}}\left|y_{q,k}[n]\right|^{2}\begin{array}{c}
\overset{\mathcal{H}_{k|1}}{>}\vspace{-0.4cm}\\
\underset{\mathcal{H}_{k|0}}{<}
\end{array}\gamma_{q,k}\vspace{-0.3cm}\label{eq:energy_detector_test}
\end{equation}
where $\gamma_{q,k}$ is the decision threshold of SU $q$ for the
carrier $k=1,\ldots,N$, to be chosen to meet the required false alarm
rate. Note that if the received samples can be assumed Gaussian distributed
(and uncorrelated), the above energy-based decision is optimal; otherwise
it is still a valuable choice when no a-priori information is available
on the primary signal features. Moreover, it is interesting to report
some results in \cite{SahaiHovenTandra04,AxellLarsson_jsac11}, where
it has been argued that for several models, if the probability density
functions under both hypotheses are perfectly known, energy detection
performs close to the optimal detector. The case of partial knowledge
of the probability density functions is discussed in Sec. \ref{sec:Extensions-and-Generalizations}
within the context of robust detection; see also \cite{Scutari-Pang_DSP11}. 

Invoking the central limit theorem, the random variable $D\left(\mathbf{Y}_{q,k}\right)$
can be approximated for sufficiently large $K_{q}=\tau_{q}\, f_{q}$
by a Gaussian distribution: for $i=1,2$, $\left.D\left(\mathbf{Y}_{q,k}\right)\,\right|\mathcal{H}_{k|i}\sim\mathcal{N}\left(\mu_{q,k|i},\,\sigma_{q,k|i}^{2}/K_{q}\right)$,
where\vspace{-0.5cm} 
\[
\mu_{q,k|i}\triangleq\left\{ \begin{array}{lll}
\sigma_{q,k}^{2}, &  & \mbox{if }i=0,\\
\sigma_{I_{q,k}}^{2}+\sigma_{q,k}^{2}, &  & \mbox{if }i=1,
\end{array}\right.\quad\mbox{and}\quad\sigma_{q,k|i}^{2}\triangleq\left\{ \begin{array}{lll}
{\sigma_{q,k}^{4}}, &  & \mbox{if }i=0\\
\mbox{E}\left|I_{q,k}\right|^{4}+2\sigma_{q,k}^{2}-\left(\sigma_{I_{q,k}}^{2}-\sigma_{q,k}^{2}\right)^{2}, &  & \mbox{if }i=1.
\end{array}\right.
\]
An explicit expression of the expected value $\mbox{E}|I_{q,k}|^{4}$
above depends on the PUs' signaling features. For example, if the
primary signals are drawn from a PSK modulation, then $\sigma_{q,k|1}^{2}=\sigma_{q,k}^{2}\left(2\,\sigma_{I_{q,k}}^{2}+\sigma_{q,k}^{2}\right)$
\cite{Quan-Cui-Poor-Sayed}; whereas if the PU signaling is assumed
to be Gaussian, then $\sigma_{q,k|1}^{2}=\left(\sigma_{I_{q,k}}^{2}+\sigma_{q,k}^{2}\right)^{2}$
\cite{Liang-Zeng-Peh-Hoang_TWC08}. 

The performance of the energy detector scheme are measured in terms
of the detection probability $P_{q,k}^{\text{\,{d}}}(\gamma_{q,k},\,\tau_{q}\,)\triangleq\text{{Prob}}\left\{ D(\mathbf{Y}_{q,k})\,>\gamma_{q,k}\,|\,\mathcal{H}_{1,k}\right\} $
and false alarm probability $P_{q,k}^{\,\text{{fa}}}(\gamma_{q,k},\,\tau_{q}\,)\triangleq\text{{Prob}}\{D(\mathbf{Y}_{q,k})>\gamma_{q,k}\,|\,\mathcal{H}_{0,k}\}$
that, under the assumptions above, are given by\vspace{-0.5cm}

\begin{equation}
P_{q,k}^{\,\text{{fa}}}\left(\gamma_{q,k},\,\tau_{q}\right)=\mathcal{Q}\left(\sqrt{\tau_{q}\, f_{q}}\,\dfrac{{\gamma_{q,k}\,-\mu_{q,k|0}}}{{\sigma_{q,k|0}}}\right)\,\,\,\mbox{and}\,\,\, P_{q,k}^{\text{\,{d}}}\left(\gamma_{q,k},\,\tau_{q}\right)=\mathcal{Q}\left(\sqrt{\tau_{q}\, f_{q}}\,\dfrac{{\gamma_{q,k}\,-\mu_{q,k|1}}}{{\sigma_{q,k|1}}}\right),\label{eq:pfa_and_pd}
\end{equation}
where $\mathcal{\mathcal{Q}}(x)\triangleq(1/\sqrt{{2\pi}})\int_{x}^{\infty}e^{-t^{2}/2}dt$
is the Q-function. The interpretation of $P_{q,k}^{\,\text{{fa}}}\left(\gamma_{q,k},\,\tau_{q}\right)$
and $P_{q,k}^{\text{\,{d}}}\left(\gamma_{q,k},\,\tau_{q}\right)$
within the CR scenario is the following: $1-P_{q,k}^{\,\text{{fa}}}$
signifies the probability of successfully identifying from the SU
$q$ a spectral hole over carrier $k$, whereas the missed detection
probability $P_{q,k}^{\text{\,{miss}}}\triangleq1-P_{q,k}^{\text{\,{d}}}$
represents the probability of SU $q$ failing to detect the presence
of the PUs on the subchannel $k$ and thus generating interference
against the PUs. The free variables to optimize are the detection
thresholds $\gamma_{q,k}$'s and the sensing times $\tau_{q}$'s;
ideally, we would like to choose $\gamma_{q,k}$'s and $\tau_{q}$'s
in order to minimize both $P_{q,k}^{\,\text{{fa}}}$ and $P_{q,k}^{\text{\,{miss}}}$,
but (\ref{eq:pfa_and_pd}) shows that there exists a trade-off between
these two quantities that will affect both primary and secondary performance.
It turns out that, $\gamma_{q,k}$'s and $\tau_{q}$'s can not be
chosen by focusing only on the detection problem (as in classical
decision theory), but the optimal choice of $\gamma_{q,k}$ and $\tau_{q}$
must be the result of a \emph{joint} optimization of the sensing and
transmission strategies over the two phases; such a novel optimization
is formulated in Sec. \ref{sub:-Game-theoretical-formulation}. \vspace{-0.2cm}

\subsection{The transmission phase\label{sub:The-transmission-phase}\vspace{-0.1cm}}

We model the set of the $Q$ active SUs as a frequency-selective $N$-parallel
interference channel, where $N$ is the number of available subcarriers;
no interference cancellation is performed at the secondary receivers
and the Multi-User Interference (MUI) is treated as additive colored
noise at each receiver. This model is suitable for most of the CR
scenarios, where the SUs coexisting in the network operate in a uncoordinated
manner without cooperating with each other, and no centralized authority
is assumed to handle the network access for the SUs. The transmission
strategy of each SU $q$ is then the power allocation vector $\mathbf{p}_{q}=\{p_{q,k}\}_{k=1}^{N}$
over the $N$ subcarriers ($p_{q,k}$ is the power allocated over
carrier $k$), subject to the following transmit power constraints\vspace{-0.2cm}
\begin{equation}
{\mathcal{P}}_{q}\triangleq\left\{ \mathbf{p}_{q}\in\mathbb{R}^{N}\,:\,\sum_{k=1}^{N}p_{q,k}\leq P_{q},\quad\mathbf{0}\leq\mathbf{p}_{q}\leq\mathbf{p}_{q}^{\max}\right\} ,\vspace{-0.1cm}\label{set_P_q}
\end{equation}
where we also included (possibly) local spectral mask constraints
$\mathbf{p}_{q}^{\max}=(p_{q,k}^{\max})_{k=1}^{N}$ {[}the vector
inequality in (\ref{set_P_q}) is component-wise{]}. 

\noindent \textbf{Opportunistic throughput.} The goal of each SU
is to optimize his ``opportunistic spectral utilization'' of the
licensed spectrum, \emph{having no a priori knowledge of the probabilities
on the PUs' presence}. According to this paradigm, each subcarrier
$k$ is available for the transmission of SU $q$ if no primary signal
is detected over that frequency band, which happens with probability
$1-P_{q,k}^{\,\text{{fa}}}$. This motivates the use of the \emph{aggregate
opportunistic throughput }as a measure of the spectrum efficiency
of each SU $q$. Given the power allocation profile $\mathbf{p}\triangleq(\mathbf{p}_{q})_{q=1}^{Q}$
of the SUs, the detection thresholds $\boldsymbol{{\gamma}}_{q}\triangleq\left(\gamma_{q,k}\right)_{k=1}^{N}$,
and sensing time $\tau_{q}$, the aggregate opportunistic throughput
of SU $q$ is defined as\vspace{-0.2cm}
\begin{equation}
R_{q}\left(\tau_{q},\,\mathbf{p},\,\boldsymbol{{\gamma}}_{q}\right)=\left(1-\dfrac{\tau_{q}}{T}\right)\,\sum_{k=1}^{N}\left[1-P_{q,k}^{\,\text{{fa}}}\left(\gamma_{q,k},\,\tau_{q}\right)\right]\, r_{q,k}\left(\mathbf{p}\right)\label{eq:Opportunistic-throughtput}
\end{equation}
where $1-\tau_{q}/T$ ($\tau_{q}\leq T$) is the portion of the frame
duration available for opportunistic transmissions, $P_{q,k}^{\,\text{{fa}}}\left(\gamma_{q,k},\,\tau_{q}\right)$
is the worst-case false alarm rate defined in (\ref{eq:pfa_and_pd}),
and $r_{q,k}\left(\mathbf{p}\right)$ is the maximum information rate
achievable on secondary link $q$ over carrier $k$ \emph{when no
primary signal is detected} and the power allocation of the SUs is
$\mathbf{p}$. Under basic information theoretical assumptions, the
maximum achievable rate $r_{q,k}\left(\mathbf{p}\right)$ for a specific
power allocation profile $p_{1,k},\dots,p_{Q,k}$ is\vspace{-0.2cm}
\begin{equation}
r_{q,k}\left(\mathbf{p}\right)=\log\left(1+\dfrac{|H_{qq}(k)|^{2}\, p_{q,k}}{{\sigma}_{q,k}^{2}+\sum_{r\neq q}|H_{rq}(k)|^{2}p_{r,k}}\right),\label{eq:rate_FSIC}
\end{equation}
where $\{H_{qq}(k)\}_{k=1}^{N}$ is the channel transfer function
of the direct link $q$ and $\{H_{rq}(k)\}_{k=1}^{N}$ is the cross-channel
transfer function between the secondary transmitter $r$ and the secondary
receiver $q$; and $\sigma_{q,k}^{2}$ is the variance of the background
noise over carrier $k$ at the receiver $q$ (assumed to be Gaussian
zero-mean distributed).

\smallskip{}

\noindent \textbf{Probabilistic interference constraints.} Due to
the inherent trade-off between $P_{q,k}^{\,\text{{fa}}}$ and $P_{q,k}^{\,\text{{miss}}}$
(cf. Sec. \ref{sub:Detection-problem}), maximizing the aggregate
opportunistic throughput (\ref{eq:Opportunistic-throughtput}) of
SUs will result in low $P_{q,k}^{\,\text{{fa}}}$ and thus large $P_{q,k}^{\,\text{{miss}}}$,
hence causing harmful interference to PUs (which happens with probability
$P_{q,k}^{\text{{miss}}}$). To control the interference radiated
by the SUs, we propose to impose \emph{probabilistic} interference
constraints in the form of \emph{individual }and/or \emph{global constraints}.
Individual interference constraints are imposed at the level of each
SU $q$ to limit the average interference generated at the primary
receiver, whereas global interference constraints limit the average
aggregate interference generated by all the SUs, which is in fact
the average interference experimented by the PU. Examples of such
constraints are the following. 
\begin{description}
\item [{\emph{-}}] \emph{Individual interference constraints:\vspace{-0.5cm}}
\end{description}
\begin{equation}
\sum_{k=1}^{N}P_{q,k}^{\,\text{{miss}}}\left(\gamma_{q,k},\,\tau_{q}\right)\cdot w_{q,k}\cdot p_{q,k}\leq I_{q}^{\max},\vspace{-0.3cm}\label{eq:individual_overal_interference_constraint}
\end{equation}

\begin{description}
\item [{\emph{-}}] \emph{Global interference constraints},\emph{\vspace{-0.3cm}}
\begin{equation}
\sum_{q=1}^{Q}\sum_{k=1}^{N}P_{q,k}^{\,\text{{miss}}}\left(\gamma_{q,k},\,\tau_{q}\right)\cdot w_{q,k}\cdot p_{q,k}\leq I^{\text{\ensuremath{\max}}},\vspace{-0.3cm}\label{eq:global_interference_constraints_2}
\end{equation}

\end{description}
where $I_{q}^{\max}$ {[}or $I^{\max}$ {]} is the maximum average
interference allowed to be generated by the SU $q$ {[}or all the
SU's{]}, and $w_{qk}$'s are a given set of positive weights. If an
estimate of the cross-channel transfer functions $\{G_{q}(k)\}_{k=1}^{N}$
between the secondary transmitters $q$'s and the primary receiver
is available, then the natural choice of $w_{qk}$ is $w_{qk}=|G_{q}(k)|^{2}$,
so that (\ref{eq:individual_overal_interference_constraint}) and
(\ref{eq:global_interference_constraints_2}) become the average interference
experienced at the primary receiver. In some scenarios where the primary
receivers have a fi{}xed geographical location, it may be possible
to install some monitoring devices close to each primary receiver
having the functionality of (cross-)channel measurement and estimation
of $I_{q}^{\text{{max}}}$'s. In scenarios where this option is not
feasible and the channel state information cannot be acquired, a different
choice of the weights coefficients $w_{q,k}$'s and the interference
threshold $I_{q}^{\max}$ in (\ref{eq:individual_overal_interference_constraint})
and (\ref{eq:global_interference_constraints_2}) can be made, based
e.g. on worst-case channel/interference statistics; several alternative
options are discussed in our companion paper \cite{Pang-Scutari-NNConvex_PII},
which is devoted to the design of distributed solution algorithms
along with their practical implementation.\emph{\vspace{-.2cm}}

\begin{remark}[\textbf{Individual vs. global constraints}]\rm \label{Remark_local_vs_global}The
proposed individual and/or global interference constraints provide
enough flexibility to explore the interplay between performance and
signaling among the SUs, making thus the proposed model applicable
to different CR scenarios. For instance, in the settings where the
SUs cannot exchange any signaling, the system design based on individual
interference constraints seems to be the most natural formulation
(see Sec. \ref{sub:Game-with-individual_constraints}); this indeed
leads to totally distributed algorithms with no signaling among the
SUs, as we show in the companion paper \cite{Pang-Scutari-NNConvex_PII}.
On the other hand, there are scenarios where the SUs may want to trade
some signaling for better performance; in such cases imposing global
interference constraints rather than the more conservative individual
constraints is expected to provide larger SUs' throughputs. This however
introduces a new challenge: how to enforce global interference constraints
in a distributed way? By imposing a coupling among the power allocations
of all the SUs, global interference constraints in principle would
call for a centralized optimization. A major contribution of the paper
is to propose a pricing mechanism based on the relaxation of the coupling
interference constraints as penalty term in the SUs' objective functions
(see Sec. \ref{sub:Game-with-pricing}). In the companion paper \cite{Pang-Scutari-NNConvex_PII},
we show that this formulation leads to (fairly) distributed best-response
algorithms where the SUs can update the prices via consensus based-schemes,
which requires a limited signaling among the SUs, in favor of better
performance.\end{remark}\vspace{-.5cm}

\section{Joint Optimization of Sensing and Transmissions via Game Theory\label{sub:-Game-theoretical-formulation}\emph{\vspace{-.1cm}}}

We focus now on the system design and formulate the joint optimization
of the sensing parameters and the power allocation of the SUs within
the framework of game theory. Motivated by the discussion in the previous
section (cf. Remark \ref{Remark_local_vs_global}), we propose next
two classes of equilibrium problems: i) games with individual constraints
only; and ii) games with individual \emph{and} global constraints.
The former formulation is suitable for modeling scenarios where the
SUs are selfish users who are not willing to cooperate, whereas the
latter class of games is applicable to the design of systems where
the SUs can exchange limited signaling in favor of better performance.
At the best of our knowledge, both formulations based on the \emph{joint}
optimization of sensing and transmission strategies are novel in the
literature. We introduce first the two formulations in Sec. \ref{sub:Game-with-individual_constraints}
and Sec. \ref{sub:Game-with-pricing} below, and then provide a unified
analysis of both games. \emph{\vspace{-.2cm}}

\subsection{Game with individual constraints \label{sub:Game-with-individual_constraints}\emph{\vspace{-.2cm}}}

In the proposed game, each SU is modeled as a player who aims to maximize
his own opportunistic throughput $R_{q}\left(\mathbf{p},\,\boldsymbol{{\gamma}}_{q},\,\tau_{q}\right)$
by choosing \emph{jointly} a proper power allocation strategy $\mathbf{p}_{q}=(p_{q,k})_{k=1}^{N}$,
detection thresholds $\boldsymbol{{\gamma}}_{q}\triangleq(\gamma_{q,k})_{k=1}^{N}$,
and sensing time $\tau_{q}$, subject to power and individual probabilistic
interference constraints. Stated in mathematical terms, player $q$'s
optimization problem is to determine, for given $\mathbf{p}_{-q}\triangleq((p_{r,k})_{k=1}^{N})_{q\neq r=1}^{Q}$,
a tuple $(\tau_{q},\,\mathbf{p}_{q},\,\boldsymbol{{\gamma}}_{q})$
such that

\framebox{\begin{minipage}[t]{0.93\columnwidth}%
\begin{equation}
\begin{array}{ll}
{\displaystyle {\operatornamewithlimits{\mbox{maximize}}_{\tau_{q},\mathbf{p}_{q},\boldsymbol{{\gamma}}_{q}}}} & R_{q}\left(\tau_{q},\,\mathbf{p},\,\boldsymbol{{\gamma}}_{q}\right)\vspace{-0.5cm}\\[0.25in]
\mbox{subject to} & \vspace{-.3cm}\\[5pt]
\mbox{{\bf (a)}} & \begin{array}{l}
{\displaystyle {\sum_{k=1}^{N}}\, P_{q,k}^{\text{{miss}}}(\gamma_{q,k},\tau_{q})\cdot w_{q,k}\cdot p_{q,k}\,\leq\, I_{q}^{\text{{max}}}},\\[0.3in]\end{array}\vspace{-1cm}\\[0.5in]
\mbox{{\bf (b)}} & \begin{array}{l}
P_{q,k}^{\text{{fa}}}(\gamma_{q,k},\tau_{q})\,\leq\,\beta_{q,k},\quad\mbox{and}\quad P_{q,k}^{\text{{miss}}}(\gamma_{q,k},\tau_{q})\,\leq\,\alpha_{q,k},\quad\forall k\,=\,1,\cdots,N,\\[0.15in]\end{array}\vspace{-1.2cm}\\[0.5in]
\mbox{{\bf (c)}} & \,\,{\displaystyle \mathbf{p}_{q}\in\mathcal{P}_{q}\epc\mbox{and}\epc\tau_{q}^{\min}\,\leq\,\tau_{q}\,\leq\,\tau_{q}^{\max}.}\vspace{.1cm}
\end{array}\label{eq:player q_individual_interference_constraints}
\end{equation}
\end{minipage}}\emph{\vspace{.3cm}}

In (\ref{eq:player q_individual_interference_constraints}) we also
included additional lower and upper bounds of $\tau_{q}$ satisfying
$0<\tau_{q}^{\min}<\tau_{q}^{\max}<T_{q}$ and upper bounds on detection
and missed detection probabilities $0<\alpha_{q,k}\leq1/2$ and $0<\beta_{q,k}\leq1/2$,
respectively. These bounds provide additional degrees of freedom to
limit the probability of interference to the PUs as well as to maintain
a certain level of opportunistic spectrum utilization from the SUs
{[}$1-P_{q,k}^{\text{{fa}}}\geq1-\beta_{q,k}${]}. Note that the constraints
$\alpha_{q,k}\leq1/2$ and $\beta_{q,k}\leq1/2$ do not represent
a real loss of generality, because practical CR systems are required
to satisfy even stronger constraints on false alarm and detection
probabilities; for instance, in the WRAN standard, $\alpha_{q,k}=\beta_{q,k}=0.1$
\cite{FCC_WGWRAN}.\emph{\vspace{-.2cm}}

\subsection{Game with pricing\label{sub:Game-with-pricing}\emph{\vspace{-.2cm}}}

We add now global interference constraints to the game theoretical
formulation. In order to enforce coupling constraints while keeping
the optimization as decentralized as possible, we propose the use
of a pricing mechanism through a penalty in the payoff function of
each player, so that the interference generated by all the SUs will
depend on these prices. Prices are thus addition variables to be optimized
(there is one common price associated to any of the global interference
constraints); they must be chosen so that any solution of the game
will satisfy the global interference constraints, which requires the
introduction of additional constraints on the prices, in the form
of price clearance conditions. Denoting by $\pi$ the price variable
associated with the global interference constraint (\ref{eq:global_interference_constraints_2}),
we have the following formulation. \vspace{0.2cm}

\framebox{\begin{minipage}[t]{0.93\columnwidth}%
Player $q$'s optimization problem is to determine, for given $\mathbf{p}_{-q}$
and $\pi$, a tuple $(\tau_{q},\,\mathbf{p}_{q},\,\boldsymbol{{\gamma}}_{q})$
such that\vspace{-0.2cm} 
\begin{equation}
\begin{array}{ll}
{\displaystyle {\operatornamewithlimits{\mbox{maximize}}_{\tau_{q},\mathbf{p}_{q},\boldsymbol{{\gamma}}_{q}}}} & R_{q}\left(\tau_{q},\,\mathbf{p},\,\boldsymbol{{\gamma}}_{q}\right)-\pi\cdot{\displaystyle {\displaystyle {\sum_{k=1}^{N}}\, P_{q,k}^{\text{{miss}}}(\gamma_{q,k},\tau_{q})\cdot w_{q,k}\cdot p_{q,k}\vspace{-1cm}}}\\[0.25in]
\mbox{subject to} & \mbox{constraints (a), (b), (c) as in }(\ref{eq:player q_individual_interference_constraints}).\\[5pt]
\end{array}\vspace{-0.2cm}\label{eq:player q}
\end{equation}
The price obeys the following complementarity condition:\vspace{-0.4cm}
\begin{equation}
0\,\leq\,\pi\,\perp\, I^{\text{{max}}}-{\sum_{k=1}^{N}}\,{\displaystyle {\sum_{q=1}^{Q}}{\displaystyle \, P_{q,k}^{\text{{miss}}}(\gamma_{q,k},\tau_{q})\cdot w_{q,k}\cdot p_{q,k}}\,\geq\,0}.\label{eq:price equilibrium}
\end{equation}
\end{minipage}}\vspace{0.3cm}

In (\ref{eq:price equilibrium}), the compact notation $0\leq a\perp b\geq0$
means $a\geq0$, $b\geq0$, and $a\,\cdot\, b=0$. The price clearance
conditions (\ref{eq:price equilibrium}) state that global interference
constraints (\ref{eq:global_interference_constraints_2}) must be
satisfied together with nonnegative price; in addition, they imply
that if the global interference constraint holds with strict inequality
then the price should be zero (no penalty is needed). Thus, at any
solution of the game, the optimal price is such that the interference
constraint is satisfied. In the companion paper \cite{Pang-Scutari-NNConvex_PII},
we show that a pricing-based game as (\ref{eq:player q})-(\ref{eq:price equilibrium})
can be solved using best-response iterative algorithms, where the
price $\pi$ is updated by the SUs themselves via consensus schemes;
which require a limited signaling exchange only among neighboring
nodes. In Sec. \ref{sec:Numerical-Results} we show that when the
SUs are willing to collaborate in the form described above, they reach
higher throughputs, which motivates the formulation (\ref{eq:player q})-(\ref{eq:price equilibrium}). 

Note that when the price $\pi$ is set to zero, the game (\ref{eq:player q})-(\ref{eq:price equilibrium})
reduces to the NEP in (\ref{eq:player q_individual_interference_constraints})
where only the individual interference constraints (\ref{eq:individual_overal_interference_constraint})
are imposed. Therefore, hereafter we focus only on (\ref{eq:player q})-(\ref{eq:price equilibrium}),
as a unified formulation including also the game (\ref{eq:player q_individual_interference_constraints})
as special case (when $\pi=0$). Note that when we deal with the general
formulation (\ref{eq:player q})-(\ref{eq:price equilibrium}), we
make the blanket assumption that none of the power/interference constraints
are redundant. \emph{\vspace{-.2cm}}

\subsection{Equivalent reformulation of the games\emph{ }and main notation\emph{\vspace{-.2cm}}}

Before starting the analysis of the proposed games, we rewrite the
players' optimization problems (\ref{eq:player q}) with the side
constraint (\ref{eq:price equilibrium}) in a more convenient and
equivalent form. The principal advantage of the reformulation is that
it converts the nonlinear constraints (b) and (c) into linear constraints;
thus (i) facilitating the application of the VI approach to the analysis
of the game, and (ii) shedding some light on the interpretation of
the new concepts and results that will be introduced later on. 

At the heart of the aforementioned equivalence is a proper change
of variables: instead of working with the original variables $\gamma_{q,k}$'s
and $\tau_{q}$'s, we introduce the new variables $\wh{\gamma}{}_{q,k}$'s
and $\wh{\tau}_{q}$'s, defined as
\begin{equation}
(\gamma_{q,k},\,\tau_{q})\mapsto(\wh{\gamma}_{q,k},\,\wh{\tau}_{q}):\hspace{1em}\wh{\tau}_{q}\triangleq\sqrt{\tau_{q}\, f_{q}}\,\quad\mbox{and}\quad\wh{\gamma}{}_{q,k}\triangleq\sqrt{\tau_{q}\, f_{q}}\,\dfrac{{\gamma_{q,k}\,-\mu_{q,k|0}}}{{\sigma_{q,k|0}}}.\label{eq:bijection}
\end{equation}
Note that the above transformation is a one-to-one mapping for any
$\gamma_{q,k}\geq0$ and $\tau_{q}>0$; we can thus reformulate the
games using the new variables without loss of generality. Moreover,
a property of (\ref{eq:bijection}) is that the constraints on $P_{q,k}^{\text{{fa}}}(\gamma_{q,k},\tau_{q})$
and $P_{q,k}^{\text{{miss}}}(\gamma_{q,k},\tau_{q})$ in each player's
optimization problem {[}see (\ref{eq:player q_individual_interference_constraints}){]}
become linear in the new variables $\wh{\gamma}_{q,k}$ and $\wh{\tau}_{q}$:
for each $k=1,\ldots,N$, we have
\begin{equation}
\begin{array}{lll}
P_{q,k}^{\text{{fa}}}(\gamma_{q,k},\tau_{q})\,\leq\,\beta_{q,k}, & \Leftrightarrow & \wh{\gamma}{}_{q,k}\geq\mathcal{Q}^{-1}\left(\beta_{q,k}\right)\bigskip\\
P_{q,k}^{\text{{miss}}}(\gamma_{q,k},\tau_{q})\,\leq\,\alpha_{q,k}, & \Leftrightarrow & {\displaystyle {\frac{\sigma_{{q,k}|0}\,\wh{\gamma}{}_{q,k}-(\,\mu_{{q,k}|1}-\mu_{{q,k}|0}\,)\,\wh{\tau}_{q}}{\sigma_{{q,k}|1}}}}\leq\mathcal{Q}^{-1}\left(1-\alpha_{q,k}\right),
\end{array}\label{eq:ineq_new_variables}
\end{equation}
where $\mathcal{Q}^{-1}\left(\cdot\right)$ denotes the inverse of
the Q-function {[}$\mathcal{Q}(x)$ is a strictly decreasing function
on $\mathbb{R}${]}. Using the above transformation, we can equivalently
rewrite the false-alarm rate $P_{q,k}^{\text{{fa}}}(\gamma_{q,k},\tau_{q})$,
the missed detection probability $P_{q,k}^{\text{{miss}}}(\gamma_{q,k},\tau_{q})$,
and the throughput $R_{q}(\tau_{q},\,\mathbf{p},\,\boldsymbol{{\gamma}}_{q})$
of each player $q$ in terms of the new variables $\wh{\boldsymbol{{\gamma}}}_{q}\triangleq(\wh{\gamma}_{q,k})_{k=1}^{N}$'s
and $\wh{\tau}_{q}$'s, denoted by $\wh{P}_{q,k}^{\text{{fa}}}(\wh{\gamma}_{q,k},\tau_{q})$,
$\wh{P}_{q,k}^{\text{{miss}}}(\wh{\gamma}_{q,k},\tau_{q})$, and $\hat{R}_{q}(\wh{\tau}_{q},\,\mathbf{p},\,\wh{\boldsymbol{{\gamma}}}_{q})$,
respectively; the explicit expression of these quantities is:\vspace{-0.3cm}
\begin{equation}
P_{q,k}^{\text{{fa}}}(\gamma_{q,k},\tau_{q})=\wh{P}_{q,k}^{\text{{fa}}}(\wh{\gamma}_{q,k})\triangleq\mathcal{Q}(\wh{\gamma}_{q,k})\label{eq:P_fa_new}
\end{equation}
\begin{equation}
P_{q,k}^{\text{{miss}}}(\gamma_{q,k},\tau_{q})=\wh{P}_{q,k}^{\text{{miss}}}(\wh{\gamma}_{q,k},\wh{\tau}_{q})\triangleq\mathcal{Q}\left({\displaystyle {\frac{\sigma_{{q,k}|0}\,\wh{\gamma}{}_{q,k}-(\,\mu_{{q,k}|1}-\mu_{{q,k}|0}\,)\,\wh{\tau}_{q}}{\sigma_{{q,k}|1}}}}\right)\label{eq:P_miss_new}
\end{equation}
\begin{equation}
R_{q}(\tau_{q},\,\mathbf{p},\,\boldsymbol{{\gamma}}_{q})=\hat{R}_{q}(\wh{\tau}_{q},\,\mathbf{p},\,\wh{\boldsymbol{{\gamma}}}_{q})\triangleq\left(\,1-{\displaystyle {\frac{\wh{\tau}_{q}^{2}}{f_{q}\, T_{q}}}\,}\right)\,{\displaystyle {\sum_{k=1}^{N}}\,\left(\,1-\wh{P}_{q,k}^{\text{{fa}}}(\wh{\gamma}_{q,k},\wh{\tau}_{q})\,\right)\, r_{q,k}\left(\mathbf{p}\right).}\label{eq:R_q_new}
\end{equation}
Using (\ref{eq:ineq_new_variables})-(\ref{eq:R_q_new}), the game
(\ref{eq:player q})-(\ref{eq:price equilibrium}) in the original
players' variables $\left(\tau_{q},\mathbf{p}_{q},\boldsymbol{{\gamma}}_{q}\right)$'s
{[}and thus also (\ref{eq:player q_individual_interference_constraints}){]}
can be equivalently rewritten in the new variables $\left(\hat{{\tau}}_{q},\mathbf{p}_{q},\hat{{\boldsymbol{{\gamma}}}}_{q}\right)$'s
as :\vspace{-0.9cm}

\begin{center}
\framebox{\begin{minipage}[t]{1\columnwidth}%
\textbf{Players' optimization} \textbf{problems}. The optimization
problem of player $q$ is:\vspace{-0.2cm}
\begin{equation}
\begin{array}{ll}
{\displaystyle {\operatornamewithlimits{\mbox{maximize}}_{\wh{\tau}_{q},\,\mathbf{p}_{q},\,\hat{{\boldsymbol{{\gamma}}}}_{q}}}} & \hat{R}_{q}\left(\wh{\tau}_{q},\,\mathbf{p},\,\hat{{\boldsymbol{{\gamma}}}}_{q}\right)-{\displaystyle {\displaystyle \pi\,\cdot\sum_{k=1}^{N}\wh{P}_{q,k}^{\text{{miss}}}\left(\wh{\gamma}_{q,k},\wh{\tau}_{q}\right)\cdot w_{q,k}\cdot p_{q,k}\vspace{-0.6cm}}}\\[0.25in]
\mbox{subject to}\vspace{-0.2cm}\\[5pt]
\mbox{{\bf (\ensuremath{\wh{\mathbf{a}}})}} & \begin{array}{l}
{\displaystyle {\sum_{k=1}^{N}}\,\wh{P}_{q,k}^{\text{{miss}}}\left(\wh{\gamma}_{q,k},\wh{\tau}_{q}\right)\cdot w_{q,k}\cdot p_{q,k}\,\leq\,{I}_{q}^{\text{\ensuremath{\max}}}}\\[0.3in]\end{array}\vspace{-0.8cm}\\[0.5in]
\mbox{{\bf (\ensuremath{\wh{\mathbf{b}}})}} & \left\{ \wh{\gamma}{}_{q,k}\geq\,\wh{\beta}_{q,k},\quad\dfrac{\sigma_{{q,k}|0}\,\wh{\gamma}{}_{q,k}-(\,\mu_{{q,k}|1}-\mu_{{q,k}|0}\,)\,\wh{\tau}_{q}}{\sigma_{{q,k}|1}}\leq\wh{\alpha}_{q,k}\right\} \quad\forall k=1,\ldots,N,\vspace{-0.8cm}\\[0.5in]
\mbox{{\bf (\ensuremath{\wh{\mathbf{c}}})}} & {\displaystyle \mathbf{p}_{q}\in\mathcal{P}_{q}\epc\mbox{and}\epc{\wh{\tau}_{q}^{\,\min}}\,\leq\,\wh{\tau}_{q}\,\leq\,{\wh{\tau}_{q}^{\,\max}},\smallskip}
\end{array}\label{eq:player q transformed_2}
\end{equation}
\textbf{Price equilibrium}. The price obeys the following complementarity
condition:\vspace{-0.8cm}

\begin{flushleft}
\begin{equation}
0\,\leq\,\pi\,\perp\, I^{\text{{max}}}-{\sum_{k=1}^{N}}\,{\displaystyle {\sum_{q=1}^{Q}}{\displaystyle \,\wh{P}_{q,k}^{\text{{miss}}}\left(\wh{\gamma}_{q,k},\wh{\tau}_{q}\right)\cdot w_{q,k}\cdot p_{q,k}}\,\geq\,0}\label{eq:price equilibrium_new}
\end{equation}

\par\end{flushleft}%
\end{minipage}}
\par\end{center}

\noindent where $\wh{\alpha}_{q,k}\triangleq\mathcal{Q}^{-1}(1-\alpha_{q,k})$,
 $\wh{\beta}_{q,k}\triangleq\mathcal{Q}^{-1}(\beta_{q,k})$, $\wh{\tau}_{q}^{\,\min}\triangleq\sqrt{f_{q}\,\tau_{q}^{\min}}$,
and $\wh{\tau}_{q}^{\,\max}\triangleq\sqrt{f_{q}\,\tau_{q}^{\max}}$;
and ($\wh{a}$), ($\wh{b}$), and ($\wh{c}$) in (\ref{eq:player q transformed_2})
are the image under the transformation (\ref{eq:bijection}) of (a),
(b), and (c) in (\ref{eq:player q_individual_interference_constraints}),
respectively {[}see (\ref{eq:bijection}){]}. 

Note that (even in the transformed domain), players' optimization
problems in (\ref{eq:player q transformed_2}) are still nonconvex,
with the nonconvexity occurring in the objective function and the
interference constraints ($\wh{a}$); the constraints ($\wh{b}$)
and ($\wh{c}$) are instead linear and thus convex; the polyhedra
in ($\wh{b}$) may be however empty. To explore such a structure,
as final step, we rewrite the game (\ref{eq:player q transformed_2})-(\ref{eq:price equilibrium_new})
in a more compact form, making explicit in the feasible set of each
player the polyhedral (convex) part {[}constraints ($\wh{b}$) and
($\wh{c}$){]} and the nonconvex part {[}constraint ($\wh{a}$){]}.
Let us denote by $\mathcal{X}_{q}$ the feasible set of player $q$
in (\ref{eq:player q transformed_2}), given by\vspace{-0.2cm}
\begin{equation}
\mathcal{X}_{q}\triangleq\left\{ (\wh{\tau}_{q},\,\mathbf{p}_{q},\,\wh{\boldsymbol{{\gamma}}}_{q})\in\mathcal{Y}_{q}\,\,|\,\, I_{q}(\wh{{\tau}_{q}},\,\mathbf{p}_{q},\,\wh{\boldsymbol{{\gamma}}}_{q})\leq0\right\} \vspace{-0.2cm}\label{eq:set_Xq}
\end{equation}
where we have separated the convex part and the nonconvex part; the
convex part is given by the polyhedron $\mathcal{Y}_{q}$ corresponding
to the constraints ($\wh{b}$) and ($\wh{c}$)\vspace{-0.2cm}
\begin{equation}
\hspace{-1em}\mathcal{Y}_{q}\triangleq\left\{ \begin{array}{ll}
(\wh{\tau}_{q},\,\mathbf{p}_{q},\,\wh{\boldsymbol{{\gamma}}}_{q})\,|\, & \wh{\gamma}_{q,k}\,\geq\,\wh{\beta}_{q,k},\quad{\displaystyle {\frac{{\sigma}_{{q,k}|0}\,\wh{\gamma}_{q,k}-(\,{\mu}_{{q,k}|1}-{\mu}_{{q,k}|0}\,)\,\wh{\tau}_{q}}{{\sigma}_{{q,k}|1}}}\,\leq\,\wh{\alpha}_{q,k}},\quad\forall k=1,\ldots,N\\
 & \mathbf{p}_{q}\in\mathcal{P}_{q},\qquad\epc\epc\wh{\tau}_{q}^{\,\min}\,\leq\,\wh{\tau}_{q}\,\leq\,\wh{\tau}_{q}^{\,\max}
\end{array}\right\} ,\vspace{-0.2cm}\label{eq:def_Y_q}
\end{equation}
whereas the nonconvex part is given by the constraint ($\wh{a}$)
that we have rewritten introducing the local interference violation
function\vspace{-0.3cm} 
\begin{equation}
I_{q}(\wh{{\tau}_{q}},\,\mathbf{p}_{q},\,\wh{\boldsymbol{{\gamma}}}_{q})\triangleq{\displaystyle {\sum_{k=1}^{N}}\,\wh{P}_{q,k}^{\text{{miss}}}\left(\wh{\gamma}_{q,k},\wh{\tau}_{q}\right)\cdot w_{q,k}\cdot p_{q,k}\,-\,{I}_{q}^{\text{\ensuremath{\max}}}},\vspace{-0.2cm}\label{eq:I_q}
\end{equation}
which measure the violation of the \emph{local} interference constraint
($\wh{a}$) at $(\wh{\tau}_{q},\,\mathbf{p}_{q},\,\wh{\boldsymbol{{\gamma}}}_{q})$.
Similarly, it is convenient to introduce also the \emph{global} interference
violation function $I(\wh{\boldsymbol{{\tau}}},\,\mathbf{p},\,\wh{\boldsymbol{{\gamma}}})$:\vspace{-0.4cm}

\begin{equation}
I(\wh{\boldsymbol{{\tau}}},\,\mathbf{p},\,\wh{\boldsymbol{{\gamma}}})\triangleq{\sum_{k=1}^{N}}\,{\displaystyle {\sum_{q=1}^{Q}}{\displaystyle \,\wh{P}_{q,k}^{\text{{miss}}}\left(\wh{\gamma}_{q,k},\wh{\tau}_{q}\right)\cdot w_{q,k}\cdot p_{q,k}}\,-I^{\text{{max}}}},\vspace{-.2cm}\label{eq:map_interference}
\end{equation}
which measures the violation of the \emph{shared} constraint at $(\wh{\boldsymbol{{\tau}}},\,\mathbf{p},\,\wh{\boldsymbol{{\gamma}}})$. 

Based on the above definitions, throughout the paper, we will use
the following notation. The convex part of the \emph{joint} strategy
set is denoted by $\mathcal{Y}\triangleq\prod_{q=1}^{Q}\mathcal{Y}_{q}$,
whereas the set containing all the (convex part of) players' strategy
sets except the $q$-th one is denoted by $\mathcal{Y}_{-q}\triangleq\prod_{r\neq q}\mathcal{Y}_{r}$;
similarly, we define $\mathcal{X}\triangleq\prod_{q=1}^{Q}\mathcal{X}_{q}$
and $\mathcal{X}_{-q}\triangleq\prod_{r\neq q}\mathcal{X}_{r}$. For
notational simplicity, when it is needed, we will use interchangeably
either $\left(\wh{\boldsymbol{{\tau}}}_{q},\,\mathbf{p}_{q},\,\wh{\boldsymbol{{\gamma}}}_{q}\right)$
or $\mathbf{x}_{q}\triangleq\left(\wh{\boldsymbol{{\tau}}}_{q},\,\mathbf{p}_{q},\,\wh{\boldsymbol{{\gamma}}}_{q}\right)$
to denote the strategy tuple of player $q$; similarly, $\mathbf{x}\triangleq\left(\wh{\boldsymbol{{\tau}}},\,\mathbf{p},\,\wh{\boldsymbol{{\gamma}}}\right)=(\mathbf{x}_{q})_{q=1}^{Q}$
will denote the strategy profile $\left(\wh{\boldsymbol{{\tau}}},\,\mathbf{p},\,\wh{\boldsymbol{{\gamma}}}\right)$
of all the players, with $\wh{\boldsymbol{{\tau}}}\triangleq(\wh{\tau}{}_{q})_{q=1}^{Q}$,
$\mathbf{p}\triangleq(\mathbf{p}_{q})_{q=1}^{Q},$ and $\wh{\boldsymbol{{\gamma}}}\triangleq(\wh{\boldsymbol{{\gamma}}}_{q})_{q=1}^{Q}$,
whereas the strategy profile of all the players except the $q$-th
one will be $\mathbf{x}_{-q}\triangleq\left(\wh{\boldsymbol{{\tau}}}_{r},\,\mathbf{p}_{r},\,\wh{\boldsymbol{{\gamma}}}_{r}\right)_{r=1,r\neq q}^{Q}$.
Using the above definitions, game (\ref{eq:player q transformed_2})-(\ref{eq:price equilibrium_new})
can be rewritten as\vspace{-0.9cm}

\begin{center}
\framebox{\begin{minipage}[t]{1\columnwidth}%
\textbf{Players' optimization}. The optimization problem of player
$q$ is: 
\begin{equation}
\begin{array}{lll}
\underset{\mathbf{x}_{q}}{\mbox{maximize}} &  & \hat{R}_{q}\,(\mathbf{x}_{q},\,\mathbf{x}_{-q})-{\pi}\cdot I(\mathbf{x}_{q},\,\mathbf{x}_{-q})\\
\mbox{subject to} &  & \mathbf{x}_{q}\triangleq\left(\wh{\boldsymbol{{\tau}}}_{q},\,\mathbf{p},\,\wh{\boldsymbol{{\gamma}}}_{q}\right)\in\mathcal{X}_{q}.
\end{array}\vspace{-0.1cm}\label{eq:player q transformed 1}
\end{equation}
\textbf{Price equilibrium}. The price obeys the following complementarity
condition:\vspace{-0.5cm}

\begin{equation}
0\,\leq\,{\pi}\,\perp\,-I\left(\mathbf{x}\right)\geq0.\label{eq:price equilibrium_vector_form}
\end{equation}
\end{minipage}}
\par\end{center}

Given the equivalence between (\ref{eq:player q transformed_2})-(\ref{eq:price equilibrium_new})
and (\ref{eq:player q transformed 1})-(\ref{eq:price equilibrium_vector_form}),
in the following we will focus on the game in the form (\ref{eq:player q transformed 1})-(\ref{eq:price equilibrium_vector_form})
w.l.o.g.. For future convenience, Table \ref{table_notation} collects
the main definitions and symbols used in (\ref{eq:player q transformed 1})-(\ref{eq:price equilibrium_vector_form}).
\vspace{-0.3cm}\hspace{-0.4cm}

\begin{table}[ht]\vspace{-0.1cm} 
\caption{Glossary of the  main notation of  game (\ref{eq:player q transformed 1})-(\ref{eq:price equilibrium_vector_form})} \vspace{0.3cm} 
\centering 
\vline\vline
\begin{tabular}{l l } 
\hline\hline                  
\hspace{1.5cm}Symbol & \hspace{3.2cm}Meaning \\
\hline\vspace{-0.3cm}\\
$\tau_q$ &   sensing time of SU $q$    \\ 
${\boldsymbol{\gamma}}_q=({\gamma}_{q,k})$ & detection thresholds  of SU $q$ over carriers $k=1,\ldots,N$ \\
$\mathbf{p}_q\triangleq (p_{q,k})_{k=1}^N$ & power allocation vector of SU $q$   \\$\pi$ & scalar price variable \\
$\wh{\tau}_q$ &   (transformed) sensing time of SU $q$  [(\ref{eq:bijection})]\\
$\wh{\boldsymbol{\gamma}}_q=(\wh{\gamma}_{q,k})$ & (transformed) detection thresholds  of SU $q$ over carriers $k=1,\ldots,N$  [(\ref{eq:bijection})]\\
$\wh{P}_{q,k}^{\text{miss}}(\wh{\tau}_q,\wh{\boldsymbol{\gamma}}_q)$ & (transformed) missed detection probability of SU $q$ on carrier $k$ [cf. (\ref{eq:P_miss_new})]\\
$\mathbf{x}_{q}\triangleq(\wh{\tau}_{q},\,\mathbf{p}_{q},\, \boldsymbol{\gamma}_q)$ & strategy tuple of SU $q$ (in the transformed variables)\\
$\mathbf{x}_{-q}\triangleq (\wh{\tau}_{r},\,\mathbf{p}_{r},\, \boldsymbol{\gamma}_r)_{r\neq q}$ & strategy profile of all the SUs (in the transformed variables) except the $q$-th one\\
$\mathbf{x}\triangleq (\mathbf{x}_q)_{q=1}^{Q}=(\wh{\boldsymbol{{\tau}}},\,\mathbf{p},\,\boldsymbol{\gamma})$ & strategy profile of all the SUs (in the transformed variables) \\
$\hat{R}_{q}\,(\mathbf{x}_{q},\,\mathbf{x}_{-q})$ & opportunistic throughput of SU $q$ in the transformed variables [cf. (\ref{eq:R_q_new})]\\
$I_q(\mathbf{x}_q)$ & local interference constraint violation of SU $q$ [cf. (\ref{eq:I_q})]\\
$I(\mathbf{x})$ & global interference constraint violation of SU $q$ [cf. (\ref{eq:map_interference})]\\
${\cal{X}}_q$,  $\mathcal{X}\triangleq\prod_{q=1}^{Q}\mathcal{X}_{q}$ & feasible set of SU $q$   [cf. (\ref{eq:set_Xq})],  joint feasible strategy set\\
$\mathcal{X}_{-q}\triangleq\prod_{r\neq q}\mathcal{X}_{r}$ & joint strategy set of the SUs except the $q$-th one\\
${\cal{Y}}_q$, $\mathcal{Y}\triangleq\prod_{q=1}^{Q}\mathcal{Y}_{q}$ & convex part of ${\cal{X}}_q$ [cf. (\ref{eq:def_Y_q})], Cartesian product of all ${\cal{Y}}_q$'s \vspace{0.2cm}\\
\hline\hline  
\end{tabular}\vline\vline
\label{table_notation} 
\end{table} \vspace{-0.2cm}

\section{Solution Analysis\label{sec:NE-LNE-QE}\vspace{-0.2cm}}

This section is devoted to the solution analysis of the game (\ref{eq:player q transformed 1})-(\ref{eq:price equilibrium_vector_form}).
It should be noted that except for the recent work \cite{PScutari10},
no existing theory to date is applicable to analyze this game. We
start our analysis by studying the feasibility of each optimization
problem in (\ref{eq:player q transformed 1}). We then introduce the
definitions of NE, LNE, and the relaxed concept QNE\emph{,} and establish
their main properties. \vspace{-0.1cm}

\subsection{Feasibility conditions \label{sub:Feasibility-conditions}}

The feasibility of the players' problems (\ref{eq:player q transformed 1})
requires that all the polyhedra $\mathcal{Y}_{q}$ defined in (\ref{eq:def_Y_q})
be nonempty; the following are necessary and sufficient conditions
for this to hold. Introducing the SNR detection $\texttt{{snr}}_{q,k}^{\text{{d}}}\triangleq{\sigma_{I_{q,k}}^{2}}/\sigma_{q,k}^{2}$
experimented by the SU $q$ over carrier $k$ and using the definitions
of ${\sigma}_{{q,k}|1}$ and ${\sigma}_{{q,k}|0}$ as given in Sec.
\ref{sub:Detection-problem}, the feasibility requirements are:\vspace{-0.2cm}
\begin{equation}
\sqrt{f_{q}{\tau}_{q}^{\max}}\geq\,{\displaystyle {\frac{\mathcal{Q}^{-1}(\wh{\beta}_{q,k})+|\mathcal{Q}^{-1}(\wh{\alpha}_{q,k})|\,\left({\sigma}_{{q,k}|1}/{\sigma}_{{q,k}|0}\right)}{\texttt{{snr}}_{q,k}^{\mbox{d}}}},\quad\forall k=1,\ldots,N,\,\,\,\forall q=1,\ldots,Q}.\label{eq:nec_suff_feasibility_cond}
\end{equation}
The conditions above have an interesting physical interpretation that
sheds light on the relationship between the achievable sensing performance
and sensing/system parameters. It quantifies the trade-off between
the sensing time (the product ``time-bandwidth'' $f_{q}{\tau}_{q}^{\max}$
of the system) and detection accuracy: the smaller the required false
alarm and missed detection probabilities, the larger the number of
samples and thus $\tau_{q}$ to be taken. \vspace{-0.1cm}

\subsection{NE and local NE\label{sub:NE-and-local}\vspace{-0.2cm}}

The definitions of NE and local NE of a game with price equilibrium
conditions as (\ref{eq:player q transformed 1})-(\ref{eq:price equilibrium_vector_form})
are the natural generalization of the same concepts introduced for
classical noncooperative games having no side constraints (see, e.g.,
\cite{Rosen_econ65,Ferrer_AniaEcLetters01,Schofield_bookChLNE}) and
are given next. \vspace{-0.2cm}

\begin{definition}\label{Def-Nash-equilibrium}A \textbf{\emph{Nash
equilibrium}} of the game (\ref{eq:player q transformed 1}) with
side constraints (\ref{eq:price equilibrium_vector_form}) is a strategy-price
tuple $\left(\mathbf{x}^{\star},{\pi}^{\star}\right)$, such that\vspace{-0.5cm}
\begin{equation}
\mathbf{x}_{q}^{\star}\,\in\,\underset{\mathbf{x}_{q}\,\in\,\mathcal{X}_{q}}{\mbox{\emph{argmax}}}{\displaystyle \,\left\{ \hat{R}_{q}(\mathbf{x}_{q},\mathbf{x}_{-q}^{\star})-{\pi}^{\star}\cdot I(\mathbf{x}_{q},\mathbf{x}_{-q}^{\star})\right\} ,\epc\forall\, q\,=\,1,\cdots,Q,}\vspace{-0.3cm}\label{eq:player q opt}
\end{equation}
 and\vspace{-0.3cm} 
\begin{equation}
0\,\leq\,{\pi}^{\star}\,\perp\,-\, I(\mathbf{x}^{\star})\geq0.\label{eq:side constraint-1}
\end{equation}
 A \textbf{\emph{local Nash equilibrium}} is a tuple $\left(\mathbf{x}^{\star},{\pi}^{\ast}\right)$
for which an open neighborhood $\mathcal{N}_{q}^{\star}$ of $\mathbf{x}_{q}^{\star}$
exists such that\vspace{-0.3cm} 
\begin{equation}
\mathbf{x}_{q}^{\star}\,\in\,\underset{\mathbf{x}_{q}\,\in\,\mathcal{X}_{q}\cap\,\mathcal{N}_{q}^{\star}}{\mbox{\emph{argmax}}}{\displaystyle \,\left\{ \hat{R}_{q}(\mathbf{x}_{q},\mathbf{x}_{-q}^{\star})-{\pi}^{\star}\cdot I(\mathbf{x}_{q},\mathbf{x}_{-q}^{\star})\right\} ,\epc\forall\, q\,=\,1,\cdots,Q,}\vspace{-0.2cm}\label{eq:player q local opt}
\end{equation}
 and (\ref{eq:side constraint-1}) holds. A NE (LNE) is said to be\emph{
trivial }if $\mathbf{p}_{q}^{\star}=\mathbf{0}$ for all $q=1,\ldots,Q$.
\end{definition}

Note that a NE is a LNE, but the converse is not necessarily true.
The feasibility of the optimization problems in (\ref{eq:player q transformed 1})
{[}cf. (\ref{eq:nec_suff_feasibility_cond}){]} is not enough to guarantee
the existence of a (L)NE. The classical case where a NE exists is
when the players\textquoteright{} objective functions (to be maximized)
are (quasi-)concave in their own variables given the other players\textquoteright{}
strategies, and the players\textquoteright{} constraint sets are compact
and convex and independent of their rivals\textquoteright{} strategies
(see, e.g., \cite{Rosen_econ65,Facchinei_Pang_VI-NE_bookCh_09});
other cases are when the game has a special structure, like potential
or supermodular games. The game (\ref{eq:player q transformed 1})
with side constraints (\ref{eq:price equilibrium_vector_form}) has
none of such properties; indeed, each player's optimization problem
in (\ref{eq:player q transformed 1}) is nonconvex with the nonconvexity
occurring in the objective function and the local/global interference
constraints; moreover, the set of feasible prices {[}satisfying (\ref{eq:price equilibrium_vector_form}){]}
is not explicitly bounded {[}note that these prices cannot be normalized
due to the lack of homogeneity in the players' optimization problem
(\ref{eq:player q transformed 1}){]}. Figure \ref{fig:nnconvex_set}
shows an example of feasible set of the detection threshold and power
allocation of a SU, for a given sensing time and $N=1$; the set is
nonconvex, implying the nonconvexity of the SUs' optimization problems
in (\ref{eq:player q transformed 1}). 

\begin{figure}
\vspace{-0.5cm}\center\includegraphics[height=6.8cm]{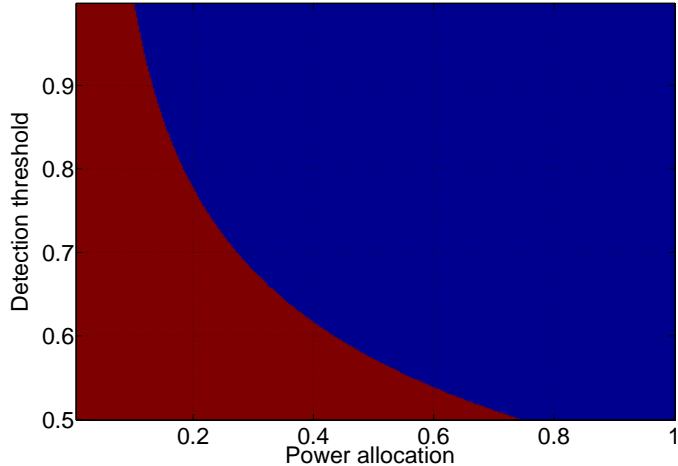}\vspace{-.4cm}

{\footnotesize \caption{{\small Example of nonconvex feasible set of a SU for fixed sensing
time and $N=1$. \label{fig:nnconvex_set}}}
}{\footnotesize \par}

\vspace{-0.2cm}
\end{figure}
 The existence of a (L)NE under the nonconvexity of the players' optimization
problems and the unboundedness of the price variables is in jeopardy;
in fact, the game may not even have a NE, or the NE may exist under
some abstract technical conditions (that are not easy to be checked%
\footnote{Conditions for the existence of a NE are established in \cite[Th. 1]{Baye-Tian-Zhou_RES93}
for NEPs (\emph{with no pricing}), which may have discontinuous and/or
nonquasi-concave payoff functions, but satisfying the so-called ``diagonal
transfer quasi-concavity'' and ``diagonal transfer continuity'';
the case of noncompact strategy sets have been addressed in \cite[Th. 5]{RePEc:ies:wpaper:e200813}
(see also references therein). Despite their theoretical importance,
the verification of these conditions for payoff functions and sets
arising from realistic applications such those in this paper does
not appear possible. %
}) requiring a particular profile of the system parameters (e.g., channel
conditions, interference level, SNR detection, etc...). This would
make the system behavior (e.g., equilibrium existence versus nonexistence,
convergence of algorithms versus nonconvergence) unpredictable, which
is not a desirable feature in practical applications. 
To overcome this issue, in this paper, we propose the use of a (relaxed)
equilibrium concept$-$the QNE\emph{$-$}which will be show to\emph{
alway} exist (even when the NE fails). The QNE along with its main
properties and optimality interpretation is introduced next. \vspace{-0.2cm}

\subsection{Quasi-Nash equilibrium\label{sub:Quasi-Equilibrium}\vspace{-0.2cm}}

For a nonlinear program constrained by finitely many algebraic equations
and inequalities and a differentiable objective function, stationarity
is defined by the well-known Karush-Kuhn-Tucker (KKT) conditions that
are necessarily satisfied by a locally optimal solution under an appropriate
Constraint Qualification (CQ), see \cite[Prop. 1.3.4]{Facchinei-Pang_FVI03};
solutions of the KKT system are called stationary solutions of the
associated optimization problem. In the context of nonconvex problems,
the common approach well-accepted in practice, is then to look for
a stationary (possibly locally optimal) solution. Here we apply this
approach to nonconvex games and introduce the concept of QNE of the
nonconvex game (\ref{eq:player q transformed 1})-(\ref{eq:price equilibrium_vector_form}),
defined as a stationary solution of the game, which is a tuple that
satisfies the KKT conditions of all the players' optimization problems
(\ref{eq:player q transformed 1}) along with the price equilibrium
constraints (\ref{eq:price equilibrium_vector_form}). The prefix
``quasi'' is intended to signify that a NE (if it exists) must be
a QNE under mild CQs.%
\footnote{Note that without some CQs being satisfied, the KKT conditions may
not be even valid \emph{necessary} conditions of optimality for the
optimization problems (\ref{eq:player q transformed 1}), making the
QNE meaningless.%
} A direct study of the main properties of the QNE (e.g., existence
and uniqueness) based on the KKT conditions of the game is not an
easy task; thus to simplify the analysis, we rewrite first the aforementioned
KKT conditions as a proper VI problem \cite{Facchinei-Pang_FVI03};
then, building on VI tools, we provide a satisfactory characterization
of the QNE. It should be noted that unlike a single optimization problem
where a stationary solution must exist (under a CQ) if an optimal
solution exists, this is not the case for a game because a QNE is
the result of concatenating the KKT conditions of several optimization
problems. 

At the basis of our approach there is an equivalent and nontrivial
reformulation of the necessary conditions for a tuple $(\mathbf{x}^{\star},\pi^{\star})$
to be a (L)NE of (\ref{eq:player q transformed 1})-(\ref{eq:price equilibrium_vector_form}),
which explores the structure of the players' feasible sets $\mathcal{X}_{q}$,
as described next. The classical approach to write the KKT conditions
of each player's optimization problem would be introducing multipliers
associated with \emph{all} the constraints in the set $\mathcal{X}_{q}$$-$both
the convex part $\mathcal{Y}_{q}$ and the nonconvex part $I_{q}(\mathbf{x}_{q})\leq0$
{[}cf. (\ref{eq:set_Xq}){]}$-$and then maximizing the resulting
Lagrangian function over\emph{ the whole }space (i.e., considering
an unconstrained optimization problem for the Lagrangian maximization).
The approach we follow here is different: instead of explicitly accounting
all the multipliers as variables of the KKT system, for each player's
optimization problem, we introduce multipliers \emph{only for the
nonconvex} constraints $\{I_{q}(\mathbf{x}_{q})\leq0,\,\, q=1,\ldots,Q\}$
{[}i.e., ($\wh{a}$){]}, and retain the convex part $\mathcal{Y}_{q}$
as explicit constraints in the maximization of the resulting Lagrangian
function. The mathematical formulation of this idea is given next.
Denoting by $\lambda_{q}$ the multiplier associated with the nonconvex
constraint $I_{q}(\mathbf{x}_{q})\leq0$ of player $q$, the Lagrangian
function of player $q$ is:\vspace{-0.2cm} 
\begin{equation}
\mathcal{L}_{q}{\displaystyle \left(\mathbf{x},\,\lambda_{q},\pi\right)}\triangleq\hat{R}_{q}(\mathbf{x}_{q},\mathbf{x}_{-q})-{\pi}\cdot I(\mathbf{x}_{q},\mathbf{x}_{-q})-\lambda_{q}\cdot I_{q}(\mathbf{x}_{q}),\label{eq:Lagrangian}
\end{equation}
which depends also on the strategies $\mathbf{x}_{-q}$ of the other
players and the price $\pi$. Building on Definition \ref{Def-Nash-equilibrium}
it is not difficult to see that if $(\mathbf{x}^{\star},\pi^{\star})$
is a (L)NE of (\ref{eq:player q transformed 1})-(\ref{eq:price equilibrium_vector_form})
and some CQ holds at $\mathbf{x}^{\star}$, there exist multipliers
$\boldsymbol{{\lambda}}^{\star}\triangleq(\lambda_{q}^{\star})_{q=1}^{Q}$
associated with the local nonconvex constraints $\{I_{q}(\mathbf{x}_{q})\leq0,\,\, q=1,\ldots,Q\}$
such that the tuple $\left(\mathbf{x}^{\star},{\pi}^{\star},\boldsymbol{{\lambda}}^{\star}\right)$
satisfies
\begin{equation}
\begin{array}{lc}
\mbox{(i)}:\quad & \mathbf{x}_{q}^{\star}\,\in\,\underset{\mathbf{x}_{q}\,\in\,\mathcal{Y}_{q}}{\mbox{argmax}}{\displaystyle \,\left\{ \mathcal{L}_{q}{\displaystyle \left(\mathbf{x}_{q},\mathbf{x}_{-q}^{\star},\lambda_{q}^{\star},\pi^{\star}\right)}\right\} ,\epc\forall q=1,\cdots,Q,}\bigskip\\
\mbox{(ii)}: & 0\,\leq\,{\pi}^{\star}\,\perp\,-\, I(\mathbf{x}^{\star})\geq0\bigskip\\
\mbox{(iii)}: & 0\,\leq\,{\lambda}_{q}^{\star}\,\perp\,-\, I_{q}(\mathbf{x}_{q}^{\star})\geq0,\epc\forall q=1,\ldots,Q.
\end{array}\label{eq:KKT_1}
\end{equation}
Note that each Lagrangian maximization in (i) is constrained over
the convex part $\mathcal{Y}_{q}$ of the player's local constraints
$\mathcal{X}_{q}$. Since $\mathcal{Y}_{q}$ is a convex set, we can
invoke the variational principle for the optimality of $\mathbf{x}_{q}^{\star}$
in (i), and obtain the following necessary conditions for (\ref{eq:KKT_1})
to hold: 
\begin{equation}
\begin{array}{lc}
\mbox{(}\mbox{i}^{'}\mbox{)}:\quad & \left(\mathbf{x}_{q}-\mathbf{x}_{q}^{\star}\right)^{T}\,\left(-\nabla_{\mathbf{x}_{q}}\mathcal{L}_{q}{\displaystyle \left(\mathbf{x}_{q},\mathbf{x}_{-q}^{\star},\lambda_{q}^{\star},\pi^{\star}\right)}\right)\geq0{\displaystyle \epc\forall\mathbf{x}_{q}\in\mathcal{Y}_{q}\,\,\,\mbox{and}\,\,\,\forall q=1,\cdots,Q,}\bigskip\\
\mbox{(ii\ensuremath{^{'}})}: & (\pi-\pi^{\star})\cdot\left(-\, I(\mathbf{x}^{\star})\right)\geq0,\quad\forall\pi\in\mathbb{R}_{+},\bigskip\\
\mbox{(iii\ensuremath{^{'}})}: & (\lambda_{q}-{\lambda}_{q}^{\star})\cdot\left(-\, I_{q}(\mathbf{x}_{q}^{\star})\right)\geq0,\quad\forall\lambda_{q}\in\mathbb{R}_{+}\,\,\,\mbox{and}\,\,\,\forall q=1,\ldots,Q,
\end{array}\label{eq:KKT_2}
\end{equation}
where $\mbox{(}\mbox{i}^{'}\mbox{)}$ is just the aforementioned first-order
optimality condition of the (nonconvex) optimization problem in $\mbox{(}\mbox{i}\mbox{)}$;
and $\mbox{(}\mbox{ii}^{'}\mbox{)}$-$\mbox{(}\mbox{iii}^{'}\mbox{)}$
are equivalent to $\mbox{(}\mbox{ii}\mbox{)}$-$\mbox{(}\mbox{iii}\mbox{)}$.
Finally, since the variables in $(\mbox{i}\ensuremath{^{'}})$-$(\mbox{iii}\ensuremath{^{'}})$
are not coupled each other by any joint constraint, we can equivalently
rewrite the three separated inequalities $(\mbox{i}\ensuremath{^{'}})$-$(\mbox{iii}\ensuremath{^{'}})$
as one inequality, obtained just summing $(\mbox{i}\ensuremath{^{'}})$-$(\mbox{iii}\ensuremath{^{'}})$.
More specifically, (\ref{eq:KKT_2}) is equivalent to\vspace{-0.2cm}
\begin{equation}
\left(\begin{array}{c}
\mathbf{x}-\mathbf{x}^{\star}\smallskip\\
\pi-\pi^{\star}\smallskip\\
\boldsymbol{{\lambda}}-\boldsymbol{{\lambda}}^{\star}
\end{array}\right)^{T}\underset{\triangleq\boldsymbol{{\Psi}}\left(\mathbf{x}^{\star},{\pi}^{\star},\boldsymbol{{\lambda}}^{\star}\right)}{\underbrace{\left(\begin{array}{c}
-\left(\nabla_{\mathbf{x}_{q}}\mathcal{L}_{q}{\displaystyle \left(\mathbf{x}^{\star},{\pi}^{\star},\,\lambda_{q}^{\star},\right)}\right)_{q=1}^{Q}\smallskip\\
-\, I(\mathbf{x}^{\star})\smallskip\\
-\,\left(I_{q}(\mathbf{x}^{\star})\right)_{q=1}^{Q}
\end{array}\right)}}\geq0,\quad\forall\left(\mathbf{x},{\pi},\boldsymbol{{\lambda}}\right)\in\underset{\triangleq\mathcal{S}}{\underbrace{{\prod_{q=1}^{Q}}\,\mathcal{Y}_{q}\times\mathbb{R}_{+}\times\mathbb{R}_{+}^{Q}}}.\label{eq:VI_ref}
\end{equation}

The above system of inequalities defines the so-called VI problem
in the variables $\left(\mathbf{x},{\pi},\boldsymbol{{\lambda}}\right)$,
whose vector function is $\boldsymbol{{\Psi}}\left(\mathbf{x},{\pi},\boldsymbol{{\lambda}}\right)$
and feasible set is $\mathcal{S}$, both defined in (\ref{eq:VI_ref});%
\footnote{The VI($\mathcal{S},\boldsymbol{\Psi}$) problem is to find a point
$\mathbf{z}^{\star}\in\mathcal{S}$, the solution of the VI, such
that $(\mathbf{z}-\mathbf{z}^{\star})^{T}\boldsymbol{\Psi}(\mathbf{z}^{\star})\geq0$
for all $\mathbf{z}\in\mathcal{S}$ \cite{Facchinei-Pang_FVI03}. %
} such a VI is denoted by VI$(\mathcal{S},\boldsymbol{{\Psi}})$. It
follows from the implications (\ref{eq:KKT_1})$\Rightarrow$(\ref{eq:VI_ref})
that the VI$(\mathcal{S},\boldsymbol{{\Psi}})$ is an equivalent reformulation
of the KKT conditions of the nonconvex game (\ref{eq:player q transformed 1})-(\ref{eq:price equilibrium_vector_form}),
wherein the convex constraints $\mathcal{Y}_{q}$'s (and thus the
associated multipliers) have been absorbed in the VI set $\mathcal{S}$.
The VI$(\mathcal{S},\boldsymbol{{\Psi}})$ is indeed composed of three
sets of variables only: i) the players' decision variables $\mathbf{x}=(\wh{\boldsymbol{{\tau}}},\,\mathbf{p},\,\wh{\boldsymbol{{\gamma}}})$;
ii) the multipliers $\boldsymbol{{\lambda}}\triangleq(\lambda_{q})_{q=1}^{Q}$
associated with the local nonconvex constraints $\{I_{q}(\mathbf{x}_{q})\leq0,\,\, q=1,\ldots,Q\}$;
and iii) the price ${\pi}$; there are no multipliers associated with
the constraints $\mathcal{Y}_{q}$'s. The above discussion is made
formal in the following lemma.\vspace{-0.2cm}

\begin{lemma} \label{Lemma_VI-KKT}The KKT conditions of the game
(\ref{eq:player q transformed 1})-(\ref{eq:price equilibrium_vector_form})
are equivalent to the VI($\mathcal{S},\,\boldsymbol{\Psi}$). The
equivalence is in the following sense: \vspace{-0.2cm}
\begin{description}
\item [{-}] Suppose that $\left(\mathbf{x}^{\star},\pi^{\star},\,\boldsymbol{{\lambda}}^{\star},\,\boldsymbol{\mu}^{\star}\right)$
is a KKT solution of the game, with $\boldsymbol{\lambda}^{\star}\triangleq(\boldsymbol{\lambda}_{q}^{\star})_{q=1}^{Q}$
and $\boldsymbol{\mu}^{\star}\triangleq(\boldsymbol{\mu}_{q}^{\star})_{q=1}^{Q}$
being the multipliers associated with the players' nonconvex constraints
$\{I_{q}(\mathbf{x}_{q}^{\star})\leq0,\,\, q=1,\ldots,Q\}$ and convex
constraints $\mathcal{Y}_{q}$'s, respectively. Then, $\left(\mathbf{x}^{\star},{\pi}^{\star},\,\boldsymbol{{\lambda}}^{\star}\right)$
is a solution of the VI($\mathcal{S},\,\boldsymbol{\Psi}$).\vspace{-0.2cm}
\item [{-}] Conversely, suppose that $\left(\overline{{\mathbf{x}}},\,\overline{{\pi}},\,\overline{\boldsymbol{{\lambda}}}\right)$
is a solution of the VI($\mathcal{S},\,\boldsymbol{\Psi}$). Then,
there exists $\overline{\boldsymbol{{\mu}}}\triangleq(\overline{\boldsymbol{{\mu}}}_{q})_{q=1}^{Q}$
such that $\left(\overline{{\mathbf{x}}},\,\overline{{\pi}},\,\overline{\boldsymbol{{\lambda}}},\,\overline{\boldsymbol{{\mu}}}\right)$
is a KKT solution of the game, with $\overline{\boldsymbol{{\mu}}}\triangleq(\overline{\boldsymbol{{\mu}}}_{q})_{q=1}^{Q}$
being the multipliers associated with the players' convex constraints
$\mathcal{Y}_{q}$'s.
\end{description}
\end{lemma}

Given the connection between a VI problem and the KKT system of the
game, the QNE of the game can also be interpreted as solutions of
the VI($\mathcal{S},\,\boldsymbol{\Psi}$), which motivates the following
formal definition of QNE of the nonconvex game (\ref{eq:player q transformed 1})-(\ref{eq:price equilibrium_vector_form}).\begin{definition}The
\textbf{\emph{quasi-Nash equilibria}}\emph{ (QNE)} of the game (\ref{eq:player q transformed 1})
with side constraints (\ref{eq:price equilibrium_vector_form}) are
the solutions $\left(\mathbf{x}^{\star},{\pi}^{\star},\,\boldsymbol{{\lambda}}^{\star}\right)$
of the VI$(\mathcal{S},\boldsymbol{\boldsymbol{\Psi}})$. A QNE is
said to be\emph{ trivial }if $\mathbf{p}_{q}^{\star}=\mathbf{0}$
for all $q=1,\ldots,Q$.\end{definition}

Our concept of QNE is conceptually similar but formally different
from other forms of local equilibria introduced in the literature,
such as the critical NE \cite{Baye-Tian-Zhou_RES93} or the generalized
equilibrium \cite{CornetaCzarnecki01}. The former is indeed a feasible
strategy profile of the players at which the gradients of each player's
objective function (taken with respect to that player's strategy)
are vanishing; the vanishing property of the gradient is typically
not satisfied by solutions of \emph{constrained} optimization problems,
let alone equilibria of inequality constraint games. The latter equilibrium
concept \cite{CornetaCzarnecki01} is defined as the solution of an
abstract set value inclusion, whose fruitful application to realistic
games like those proposed in this paper is seriously questionable.
In contrast, our VI-based definition of local equilibrium opens the
way to the use of the VI machinery \cite{Facchinei-Pang_FVI03}, both
theory and methods, to successfully study the QNE of \emph{realistic}
games. The line of analysis we are going to illustrate, based on \cite{PScutari10},
is in fact a new contribution in the literature of solution analysis
of local equilibria. \medskip{}

\noindent \textbf{Optimality interpretation of the QNE}. Since the
QNE is a stationary solution of the game, it is desirable to understand
whether the QNE has some (local) optimality properties. Interestingly,
the QNE has the following equivalent interpretation: the triplet $\left(\mathbf{x}^{\star},\pi^{\star},\,\boldsymbol{{\lambda}}^{\star}\right)$
is a QNE of the game (\ref{eq:player q transformed 1})-(\ref{eq:price equilibrium_vector_form})
if and only if it is an optimal solution of the players' optimization
problems in the following sense {[}recall that $\mathbf{x}^{\star}=(\mathbf{x}_{q}^{\star})_{q=1}^{Q}$
with $\mathbf{x}_{q}^{\star}=\left(\wh{\tau}_{q}^{\star},\,\mathbf{p}_{q}^{\star},\,\wh{\boldsymbol{{\gamma}}}_{q}^{\star}\right)$,
and $\mathbf{x}_{-q}^{\star}=\left(\wh{\tau}_{r}^{\star},\,\mathbf{p}_{r}^{\star},\,\wh{\boldsymbol{{\gamma}}}_{r}^{\star}\right)_{r\neq q}${]}:
\begin{description}
\item [{(a)}] (\emph{NE power allocation for fixed optimal thresholds and
sensing times}) The tuple $\mathbf{p}^{\star}$ is a NE of the SUs'
game when $\wh{\boldsymbol{{\tau}}}=\wh{\boldsymbol{{\tau}}}^{\star}$
and $\wh{\boldsymbol{{\gamma}}}=\wh{\boldsymbol{{\gamma}}}^{\star}$;
for each $q=1,\cdots,Q$,\vspace{-0.2cm} 
\begin{equation}
\begin{array}{lll}
\mathbf{p}_{q}^{\star}\,\in & {\displaystyle {\operatornamewithlimits{\mbox{argmax}}_{\mathbf{p}_{q}}}} & \left\{ {\displaystyle \hat{R}_{q}(\wh{\tau}_{q}^{\star},\,\mathbf{p}_{q},\,\wh{\boldsymbol{{\gamma}}}_{q}^{\star},\mathbf{x}_{-q}^{\star})-{\pi}^{\star}\cdot I(\wh{\tau}_{q}^{\star},\,\mathbf{p}_{q},\,\wh{\boldsymbol{{\gamma}}}_{q}^{\star},\mathbf{x}_{-q}^{\star})}\right\} \vspace{-0.2cm}\\[0.25in]
 & \mbox{\mbox{subject to}} & \left(\wh{\tau}_{q}^{\star},\,\mathbf{p}_{q},\,\wh{\boldsymbol{{\gamma}}}_{q}^{\star}\right)\in\mathcal{X}_{q},\\[0.3in]
\end{array}\vspace{-0.4cm}\label{eq:given thresholds}
\end{equation}
\vspace{-1.2cm} 
\item [{(b)}] (\emph{Threshold optimization for fixed optimal power allocations
and sensing times}) The threshold vector $\wh{\boldsymbol{{\gamma}}}_{q}^{\star}$
is an optimal solution of the $q$-th players' optimization problem
when $\wh{\boldsymbol{{\tau}}}=\wh{\boldsymbol{{\tau}}}^{\star}$
and $\mathbf{p}=\mathbf{p}^{\star}$: for each $q=1,\cdots,Q$,\vspace{-0.2cm}
\begin{equation}
\begin{array}{lll}
\wh{\boldsymbol{{\gamma}}}_{q}^{\star}\,\in & {\displaystyle {\operatornamewithlimits{\mbox{argmax}}_{\wh{\boldsymbol{{\gamma}}}_{q}}}} & \left\{ {\displaystyle \hat{R}_{q}(\wh{\tau}_{q}^{\star},\,\mathbf{p}_{q}^{\star},\,\wh{\boldsymbol{{\gamma}}}_{q},\mathbf{x}_{-q}^{\star})-{\pi}^{\star}\cdot I(\wh{\tau}_{q}^{\star},\,\mathbf{p}_{q}^{\star},\,\wh{\boldsymbol{{\gamma}}}_{q},\mathbf{x}_{-q}^{\star})}\right\} \vspace{-0.2cm}\\[0.25in]
 & \mbox{\mbox{subject to}} & \left(\wh{\tau}_{q}^{\star},\,\mathbf{p}_{q}^{\star},\,\wh{\boldsymbol{{\gamma}}}_{q}\right)\in\mathcal{X}_{q};\\[0.3in]
\end{array}\label{eq:given rates}
\end{equation}
\vspace{-1.2cm} 
\item [{(c)}] (\emph{Sensing time optimization for fixed optimal power
allocations and thresholds}) The sensing time $\wh{\tau}_{q}^{\star}$
is an optimal solution of the $q$-th players' optimization problem
when $\wh{\boldsymbol{{\gamma}}}_{q}=\wh{\boldsymbol{{\gamma}}}_{q}^{\star}$
and $\mathbf{p}=\mathbf{p}^{\star}$: for each $q=1,\cdots,Q$,\vspace{-0.4cm}
\begin{equation}
\begin{array}{lll}
\wh{\tau}_{q}^{\star}\,\in & {\displaystyle {\operatornamewithlimits{\mbox{argmax}}_{\wh{\tau}_{q}}}} & \left\{ {\displaystyle \hat{R}_{q}(\wh{\tau}_{q},\,\mathbf{p}_{q}^{\star},\,\wh{\boldsymbol{{\gamma}}}_{q}^{\star},\mathbf{x}_{-q}^{\star})-{\pi}^{\star}\cdot I(\wh{\tau}_{q},\,\mathbf{p}_{q}^{\star},\,\wh{\boldsymbol{{\gamma}}}_{q}^{\star},\mathbf{x}_{-q}^{\star})}\right\} \vspace{-0.2cm}\\[0.25in]
 & \mbox{\mbox{subject to}} & \left(\wh{\tau}_{q},\,\mathbf{p}_{q}^{\star},\,\wh{\boldsymbol{{\gamma}}}_{q}^{\star}\right)\in\mathcal{X}_{q};\\[0.3in]
\end{array}\vspace{-0.4cm}\label{eq:given rates_sensing_time_opt}
\end{equation}
\vspace{-1.1cm} 
\item [{(d)}] (\emph{Price equilibrium}) The complementarity condition
$0\,\leq\,{\pi}^{\star}\,\perp\,-I(\mathbf{x}^{\star})\geq0$ holds;
\vspace{-0.3cm} 
\item [{(e)}] (\emph{Common multipliers for individual interference} constraints)
For each $q=1,\cdots,Q$, there exists a \emph{common} optimal multiplier
tuple $\boldsymbol{{\lambda}}^{\star}$ associated with the individual
interference constraints $\{I_{q}(\mathbf{x}_{q}^{\star})\leq0,\,\, q=1,\ldots,Q\}$
at $\mathbf{x}^{\star}$. 
\end{description}
In words, at a QNE $\left(\mathbf{x}^{\star},{\pi}^{\star},\,\boldsymbol{{\lambda}}^{\star}\right)$
we have that: i) each user unilaterally maximizes his own function
with respect to each of his own strategies $\wh{\tau}_{q}^{\star}$,
$\mathbf{p}_{q}^{\star}$, and $\wh{\boldsymbol{{\gamma}}}_{q}^{\star}$
\emph{separately}, while keeping the rivals' strategies fixed at the
optimal value {[}statements (a)-(c) above{]}; ii) the optimal price
value ${\pi}^{\star}$ satisfies the complementarity condition (\ref{eq:price equilibrium_vector_form})
{[}statement (d) above{]}; and iii) $\boldsymbol{{\lambda}}^{\star}$
is the \emph{common} optimal multiplier for all the players associated
with the individual interference constraints in (\ref{eq:given thresholds})-(\ref{eq:given rates_sensing_time_opt})
{[}statement (e) above{]}. Note that (\ref{eq:given thresholds})
is a linearly constrained concave maximization problem; thus multipliers
exist for this problem. Problems (\ref{eq:given rates}) and (\ref{eq:given rates_sensing_time_opt})
are concave maximization program with convex constraints fulfilling
the Slater CQ for a fixed but arbitrary nonzero $\mathbf{p}_{q}^{\star}$.
Thus they also have constraint multipliers. Moreover, the constraint
sets of problems (\ref{eq:given thresholds}), (\ref{eq:given rates}),
and (\ref{eq:given rates_sensing_time_opt}) are bounded. A key requirement
in the QNE definition is therefore condition (e) that stipulates the
existence of a set of \emph{common} multiplier tuple $\boldsymbol{{\lambda}}^{\star}$
for the individual interference constraints, which is not automatically
guaranteed. Sec. \ref{sec:Existence-of-QE} focuses on this issue.\vspace{-0.2cm}

\subsection{Connection between LNE and QNE\label{sub:Connection-between-LNE_QE}\vspace{-0.1cm}}

First of all, note that a (L)NE $(\mathbf{x}^{\star},\pi^{\star})$
is composed of the sensing/transmission strategies of the players
as well as the price, whereas a QNE is a tuple $\left(\mathbf{x}^{\star},\,{\pi}^{\star},\,\boldsymbol{\lambda}^{\star}\right)$
consisting also of the KKT multipliers $\boldsymbol{{\lambda}}^{\star}$
of the game's nonconvex constraints $\{I_{q}(\mathbf{x}_{q})\leq0,\,\, q=1,\ldots,Q\}$.
To explore the connection between a LNE and a QNE we need to show
that, under the feasibility and solvability of all the players' optimization
problems in (\ref{eq:player q transformed 1}) (cf. Sec. \ref{sub:Feasibility-conditions}),
the KKT conditions of the game are valid \emph{necessary} conditions
of optimality for these problems. To this end, we need to verify that
an appropriate CQ holds. In this paper, we will use the Abadie's CQ
(ACQ) \cite[Sec. 3.2]{Facchinei-Pang_FVI03}. Not explicitly mentioned,
this CQ is essentially the key to show the validity of the following
result, which states that every LNE must be a QNE in this game.

\begin{proposition} \label{pr:LNE implies QE} Given the game (\ref{eq:player q transformed 1})
with (possibly) side constraints (\ref{eq:price equilibrium_vector_form}),
the following holds: if $(\mathbf{x}^{\star},{\pi}^{\star})$ is a
LNE of the game then $\boldsymbol{{\lambda}}^{\star}$ exists such
that $(\mathbf{x}^{\star},\pi^{\star},\,\boldsymbol{{\lambda}}^{\star})$
is a QNE. \end{proposition}

Note that the converse of Proposition \ref{pr:LNE implies QE} in
general is not true; for a QNE to be a LNE, one needs appropriate
second-order sufficiency conditions. The analysis is quite involved
and goes beyond the scope of this paper; we refer the interested reader
to \cite[Proposition 7]{PScutari10} for details. In the companion
paper \cite{Pang-Scutari-NNConvex_PII}, we derive sufficient conditions
for (a special case of) the game to have a unique QNE, which then
must coincide with a (L)NE. Finally, recall that, 
since the set $\mathcal{S}$ of the VI$(\mathcal{S},\,\boldsymbol{{\Psi}})$
is unbounded, the existence of a QNE (a solution of the VI) is not
automatically guaranteed. The study of existence and boundedness of
the QNE is addressed in the next section. \vspace{-0.1cm}

\subsection{Existence of a QNE\label{sec:Existence-of-QE}\vspace{-0.2cm}}

We can now study the existence of a QNE of the game (\ref{eq:player q transformed 1})-(\ref{eq:price equilibrium_vector_form}),
via the solution analysis of the VI$(\mathcal{S},\boldsymbol{{\Psi}})$.
The theorem below states that the game (\ref{eq:player q transformed_2}),
if feasible, always admits a non trivial QNE. 
\vspace{-0.1cm}

\begin{theorem} \label{th:QE} Suppose that the set $\mathcal{S}$
in (\ref{eq:VI_ref}) is nonempty \emph{{[}}i.e., the feasibility
conditions (\ref{eq:nec_suff_feasibility_cond}) hold true\emph{{]}}.
Then, the VI$(\mathcal{S},\boldsymbol{{\Psi}})$ has a nonempty bounded
solution set; thus the game (\ref{eq:player q transformed 1}) with
side constraints (\ref{eq:price equilibrium_vector_form}) has bounded
QNEs. Moreover, every QNE (and thus LNE) is non trivial.\vspace{-0.1cm}
\end{theorem}

\begin{proof} See Appendix.\end{proof}

Note that for the game in (\ref{eq:player q_individual_interference_constraints}),
where there are no global interference constraints (i.e., ${\pi}=0$),
it is not difficult to show that every QNE (and thus LNE) is such
that $\mathbf{p}_{q}^{\star}\neq\mathbf{0}$ for \emph{all} $q=1,\ldots,Q$.
In the presence of side constraints, instead, this is not obvious,
because of the potential unboundedness of the prices ${\pi}$ (if
the interference constraints are ``too stringent'', implying large
prices, the users may not be allowed to transmit at all{]}. Interestingly,
sufficient conditions for every QNE (LNE) of the game with side constraints
(\ref{eq:price equilibrium_vector_form}) to have $\mathbf{p}_{q}^{\star}\neq\mathbf{0}$
for some/all $q=1,\ldots,Q$ can be derived (we omit the details because
of the space limitation). These conditions have a physical interpretations:
they impose an upper bound on the (normalized) secondary-to-primary
cross-channels $|G_{q}(k)|^{2}$, so that a bounded set of equilibrium
prices exists such that the interference constraints at the primary
receiver can be satisfied with nonzero transmissions of the SUs.\vspace{-0.2cm}

\section{The Equi-sensing Case\label{sec:The-Equi-sensing-Case}\vspace{-0.2cm}}

The decision model proposed so far is based on the assumption that
only the PUs' signals are involved in the detection process performed
by the SUs, implying that the SUs are somehow able to distinguish
between primary and secondary signaling. This can be naturally accomplished
if there is a \emph{common} sensing time (still to optimize) during
which \emph{all} the SUs stay silent while sensing the spectrum. However,
the joint optimization of the sensing and transmission parameters
proposed in Sec. \ref{sub:-Game-theoretical-formulation}, in general,
leads to different optimal sensing times of the SUs, implying that
some SU may start transmitting while some others are still in the
sensing phase. Since the energy based detection as proposed in Sec.
\ref{sub:Detection-problem} is not able to discriminate between sources
of received energy, this in-band interference generated by the transmitting
SUs would confuse the energy detector and thus introduce a significant
performance degradation. To overcome this issue two different directions
can be explored, as detailed next. 

A first approach could be using more sophisticated signal processing
techniques for the SUs' sensing that look into a primary signal footprint
(e.g., modulation type, data rate, pulse shaping, or other signal
feature) to improve the detector robustness, at the cost of increased
complexity. This would allow the SUs to differentiate between primary
signals, background noise, and interference and thus possibly benefit
from adaptive signal processing for canceling SU interferers. Depending
on what a priori knowledge of the primary signal is known to the SUs,
different \emph{feature} detectors can be applied under different
scenarios and complexity requirements \cite{Ma-Li-Juang_ProcIEEE09}.
For example, ATSC digital TV signal has narrow pilot for audio and
video carriers; CDMA systems have dedicated spreading codes for pilot
and synchronization channels; OFDMA packets have preambles for packet
acquisition. Under such a priori information, the optimal detection
is given by the matched filter \cite{Levy-book-Det-Est08,Ma-Li-Juang_ProcIEEE09}.
We leave the reader the easy task to extend the proposed game theoretical
formulation to the case in which the energy detector is replaced by
 the pilot-based matched filter (note that, under mild conditions,
 the performance of the matched filter are still given by (\ref{eq:pfa_and_pd}),
but with a different expression for the $\mu_{q,k|i}$'s and $\sigma_{q,k|i}$'s
\cite{Ma-Li-Juang_ProcIEEE09}). Note, however, that the better sensing
performance of the matched filter are obtained at the cost of additional
hardware complexity: the SUs would need a dedicated receiver for every
PU class.

The second approach we propose is suitable for scenarios where feature
detection is not implementable, and thus the energy detector is the
only available option (see also Sec. \ref{sec:Extensions-and-Generalizations}
for a more general energy detector-based scheme). In such a case,
the only way for the SUs to distinguish the primary from the secondary
signals is to avoid overlapping secondary transmissions during the
sensing phase. This can be done by ``forcing'' the same sensing
time for all the SUs, which still needs to be optimized. To do that
while keeping the distributed nature of the optimization, we propose
to modify the original games as follows. \vspace{0.3cm}

\hspace{-0.7cm}%
\framebox{\begin{minipage}[t]{1\columnwidth}%
\textbf{Players' optimization}. The optimization problem of player
$q$ is: given $c>0$\vspace{-0.3cm} 
\begin{equation}
\begin{array}{lll}
\underset{\mathbf{x}_{q}}{\mbox{maximize}} &  & \hat{R}_{q}(\mathbf{x}_{q},\mathbf{x}_{-q})-{\pi}\cdot I(\mathbf{x})-\,\dfrac{{c}}{2}\,\left(\dfrac{{\hat{\tau}_{q}}}{\sqrt{{f_{q}}}}-\dfrac{{1}}{Q}\,{\displaystyle {\sum_{r=1}^{Q}}}\,\dfrac{{\hat{\tau}_{r}}}{\sqrt{{f_{r}}}}\right)^{2}\\
\mbox{subject to} &  & \mathbf{x}_{q}\triangleq\left(\wh{\tau}_{q},\,\mathbf{p}_{q},\,\wh{\boldsymbol{{\gamma}}}_{q}\right)\in\mathcal{X}_{q}.
\end{array}\vspace{-0.1cm}\label{eq:penalized_game}
\end{equation}
\textbf{Price equilibrium}. The price obeys the complementarity condition
(\ref{eq:price equilibrium_vector_form}).%
\end{minipage}}\vspace{0.4cm}

The difference with respect to the previous formulations is that now
in the objective function of each player there is an additional term
that works like a penalization in using different sensing times for
the players. Because of this penalization, one would expect that,
for sufficiently large $c$, the equilibrium of the game tends to
have equal (normalized) sensing times ${\hat{\tau}_{q}}/{\sqrt{{f_{q}}}}$'s
(and thus equal $\tau_{q}$'s, since ${\hat{\tau}_{q}}/{\sqrt{{f_{q}}}}=\sqrt{\tau_{q}}$;
see (\ref{eq:bijection}){]}, provided that such a common value is
feasible for all the players' optimization problems; this solution
indeed is the one that minimizes the loss induced by the penalization
in the payoff function of each player. This intuition is formalized
in Theorem \ref{Prop_equi_sensing_QE} below. \smallskip

\noindent \textbf{Feasibility conditions.} The first step is to derive
(sufficient) conditions guaranteeing the existence of a common value
for the sensing times $\tau_{q}$'s. We have the following: For every
$\wh{\boldsymbol{{\gamma}}}=(\wh{\boldsymbol{{\gamma}}}_{q})_{q=1}^{Q}$
for which there exist a $\mathbf{p}=(\mathbf{p}_{q})_{q=1}^{Q}$ and
$\wh{\boldsymbol{{\tau}}}=(\wh{\tau}_{q})_{q=1}^{Q}$ satisfying the
feasibility conditions ($\wh{\mbox{b}}$) and ($\wh{c}$) of the original
game in (\ref{eq:player q transformed_2}), there must exist a common
$\tau$ such that, for all $q=1,\ldots,Q$ and $k=1,\ldots N$, we
have
\begin{equation}
\dfrac{{\hat{\tau}_{q}}^{\min}}{\sqrt{{f_{q}}}}\leq\sqrt{\tau}\leq\dfrac{{\hat{\tau}_{q}}^{\max}}{\sqrt{f_{q}}},\quad\mbox{and}\quad{\frac{\sigma_{{q,k}|0}\,\wh{\gamma}_{q,k}-(\,\mu_{{q,k}|1}-{\mu}_{{q,k}|0}\,)\,\sqrt{f_{q}\,\tau}}{{\sigma}_{{q,k}|1}}}\,\leq\,\wh{\alpha}_{q,k}.\label{eq:feasibility_G_c}
\end{equation}
The first set of conditions in (\ref{eq:feasibility_G_c}) simply
postulates the existence of an overlap among the (normalized) sensing
time intervals $[{\hat{\tau}_{q}}^{\min}/\sqrt{{f_{q}}},\,{\hat{\tau}_{q}}^{\max}/\sqrt{{f_{q}}}]$,
which is necessary to guarantee the existence of a common value for
the sensing times in the original variables $\tau_{q}$'s. The second
set of conditions guarantee the existence of a common $\tau$ for
the values of feasible $\wh{\boldsymbol{{\gamma}}}$ that are candidates
to be QNE of the original game (\ref{eq:player q transformed_2}).
Note that this is less than requiring $\tau$ to satisfy also the
interference constraints ($\wh{\mbox{a}}$) in (\ref{eq:player q transformed_2})
for all feasible $\wh{\boldsymbol{{\gamma}}}$ and $\mathbf{p}$.\vspace{-0.1cm}

\begin{theorem}\label{Prop_equi_sensing_QE}Given the game in (\ref{eq:penalized_game})
with side constraints (\ref{eq:price equilibrium_vector_form}), suppose
that the feasibility conditions in (\ref{eq:feasibility_G_c}) are
satisfied. Let $\{c^{\,\nu}\}_{\nu=1}^{\infty}$ be any sequence of
positive scalars such that $\lim_{\nu\rightarrow+\infty}c^{\,\nu}=+\infty$,
and let $(\mathbf{x}^{\,\nu}\,\pi^{\,\nu},\,\boldsymbol{{\lambda}}^{\,\nu})$
be a QNE of the game for $c=c^{\,\nu}$. Then, the sequence $\{(\mathbf{x}^{\,\nu}\,\pi^{\,\nu},\,\boldsymbol{{\lambda}}^{\,\nu})\}_{\nu=1}^{\infty}$
has a limit point, denoted by $(\mathbf{x}^{\,\infty},\,{\pi}^{\,\infty},\,\boldsymbol{{\lambda}}^{\,\infty})$,
where $\mathbf{x}^{\,\infty}\triangleq(\wh{\boldsymbol{{\tau}}}^{\,\infty},\,\mathbf{p}^{\,\infty},\,\wh{\boldsymbol{{\gamma}}}^{\,\infty})$;
for every such a limit point, the following hold: 
\begin{description}
\item [{(i)}] There exists a feasible $\tau^{\star}$ such that \vspace{-0.1cm}
\begin{equation}
\dfrac{{\hat{\tau}_{q}^{\,\infty}}}{\sqrt{{f_{q}}}}=\dfrac{{\hat{\tau}_{r}^{\,\infty}}}{\sqrt{{f_{r}}}}=\sqrt{\tau^{\star}},\quad\forall r,q=1,\ldots,Q,\,\,\mathcal{\mbox{and}}\,\, r\neq q;\label{eq:equal_sensing_times}
\end{equation}

\item [{(ii)}] $\tau^{\star}$ in (\ref{eq:equal_sensing_times}) has the
following optimality properties:{\small 
\begin{equation}
\hspace{-0.3cm}\begin{array}{lll}
{\tau}^{\star}\,= & {\displaystyle {\operatornamewithlimits{\mbox{\emph{argmax}}}_{{\tau}}}} & \left\{ {\displaystyle {\sum_{q=1}^{Q}}\,}{\displaystyle \hat{R}_{q}\left(\sqrt{{f_{q}}\,\tau},\,\mathbf{p}^{\infty},\,\wh{\boldsymbol{{\gamma}}}^{\,\infty}\right)-{\pi}^{\infty}\cdot I\left(\left(\sqrt{{f_{q}}\,\tau}\right)_{q=1}^{Q},\,\mathbf{p}^{\infty},\,\wh{\boldsymbol{{\gamma}}}^{\,\infty}\right)}\right\} \\[0.25in]
 & \mbox{\emph{subject to}} & \left(\sqrt{{f_{q}}\,\tau},\,\mathbf{p}^{\infty},\,\wh{\boldsymbol{{\gamma}}}_{q}^{\infty}\right)\in\mathcal{X}_{q},\quad\forall q=1,\ldots,Q.\\[0.3in]
\end{array}\label{eq:optimality of tau_star}
\end{equation}
}{\small \par}
\end{description}
\end{theorem}\vspace{-0.5cm}

Note that (\ref{eq:equal_sensing_times}) states that in the limit
\emph{all} the original sensing time variables $\tau_{q}$ must be
equal to $\tau^{\star}$ {[}recall that ${\hat{\tau}_{q}}/{\sqrt{{f_{q}}}}=\sqrt{\tau_{q}}$,
see (\ref{eq:bijection}){]}, implying that there always exists a
sufficiently large $c$ such that one can reach a QNE of game (\ref{eq:penalized_game})
having sensing times that differ from their average by the desired
accuracy. Moreover, such a common sensing time $\tau^{\star}$ is
optimal in the sense given by (\ref{eq:optimality of tau_star}):
${\tau}^{\star}$ is the unique maximizer of the sum (priced) throughput
of the \emph{original }game (\ref{eq:player q}), satisfying the price
and interference constraints, while keeping the players' powers, thresholds,
and price fixed to $\mathbf{p}^{\,\infty}$$,$ $\wh{\boldsymbol{{\gamma}}}_{q}^{\infty}$,
and $\pi^{\,\infty}$, respectively. In the companion paper \cite{Pang-Scutari-NNConvex_PII},
we focus on distributed algorithms to compute such limit points.\vspace{-0.3cm}

\section{Extension of the Framework\label{sec:Extensions-and-Generalizations}\vspace{-0.2cm}}

The framework presented in this paper may be extended to cover more
general settings, without affecting the validity of the obtained results.
In this section, we briefly discuss some of these extensions. \vspace{-0.4cm}

\paragraph*{Composite\emph{ }hypothesis testing and robust detection.}

The sensing model introduced in Sec. \ref{sub:Detection-problem}
can be generalized to the case of multiple active PUs, and the presence
of device-level uncertainties (e.g., uncertainty in the power spectral
density of the PUs' signals and thermal noise) as well as system level
uncertainties (e.g., lack of knowledge of the number of active PUs).
In this more general setting, there are $2^{P}$ configurations of
possibly active PUs associated to the elements of the power set $\mathcal{P}(P)$
of $\{1,\ldots,P\}$;%
\footnote{The power set of a given set $\mathcal{S}$ is the set of all subsets
of $\mathcal{S}$, including the empty set and \emph{$\mathcal{S}$
}itself.%
} we then propose to formulate the spectrum sensing problem of the
SUs as a \emph{composite} hypothesis testing, based on the following
two sets of hypotheses: for each SU $q=1,\ldots,Q,$ at carrier $k=1,\ldots,N$
and time index $n=1,2,\ldots,K_{q},$\emph{\vspace{-.3cm}} 
\begin{equation}
\begin{array}{l}
\mbox{(PU signal absent) }\,\,\mathcal{H}_{0,k}:\quad y_{q,k}[n]=w_{q,k}[n]\\
\mbox{(PU signal present) }\mathcal{H}_{1,k}:\quad y_{q,k}[n]=I_{q,k|\mathcal{K}}[n]+w_{q,k}[n],\,\,\mbox{for some }\mathcal{K}\in\mathcal{P}(P)\backslash\{{\emptyset}\}\vspace{-0.2cm}
\end{array}\label{eq:Hyp_testing_composite_general}
\end{equation}
where $\mathcal{H}_{0,k}$ represents the absence of any primary signal
over the subcarrier $k$, and $\mathcal{H}_{1,k}$ represents the
presence of (at least) one PU, with $I_{q,k|\mathcal{K}}[n]$ being
the signal received by SU $q$ over carrier $k$ due to the presence
of the active PUs, indexed by the elements of $\mathcal{K}$.

Under the assumptions that i) the noise variance over each subcarrier
is known at the receiver of each SU within a given uncertainty interval;
and ii) the SUs do not know the current number of active PUs out of
$P$ PUs (the set $\mathcal{K}$), we proved that the decision rule
(\ref{eq:energy_detector_test}), with $y_{q,k}[n]$ given in (\ref{eq:Hyp_testing_composite_general}),
is a \emph{Uniformly Most Powerful} (UMP) test \cite{Scutari-Pang_DSP11}.
Roughly speaking, this means that (\ref{eq:energy_detector_test})
is optimal for both the false alarm and detection probabilities, in
the sense that it reaches the desired false alarm rate (size of the
test) while maximizing the detection probability \emph{over the entire
range of noise uncertainty and uncertainty on the set} \emph{of the
active PUs}. In \cite{Scutari-Pang_DSP11}, we showed that the (worst-case)
performance of the proposed test are formally still given by (\ref{eq:pfa_and_pd}),
but with a different expression for the constants $\mu_{q,k|i}$'s
and $\sigma_{q,k|i}^{2}$'s. Thus, the analysis developed in the previous
sections apply also to this more general model; because of the space
limitation, we refer to \cite{Scutari-Pang_DSP11} for details.\vspace{-0.4cm}

\paragraph*{Game-theoretical formulations in the presence of multiple PUs.}

To simplify the presentation, we have thus far considered scenarios
where there is only one active PU and a single global interference
constraint in the form of (\ref{eq:global_interference_constraints_2}).
The proposed game-theoretical formulations can be readily extended
to the case of multiple active PUs and additional local/global interference
constraints. For instance, suppose that there are (at most) $P$ active
PUs along with their associated local and/or global interference constraints,
e.g., given respectively by (we write the constraints directly in
the transformed variables $\wh{\gamma}_{q,k}$ and $\wh{\tau}_{q}$):
for each $q=1,\ldots,Q$,\emph{\vspace{-.2cm}}

\emph{
\begin{equation}
I_{q}^{\,(p)}(\wh{\tau}_{q},\,\mathbf{p}_{q},\,\wh{\boldsymbol{{\gamma}}}_{q})\triangleq{\displaystyle {\sum_{k=1}^{N}}\,\wh{P}_{q,k}^{\text{{miss}}}\left(\wh{\gamma}_{q,k},\wh{\tau}_{q}\right)\cdot w_{q,k}^{\,(p)}\cdot p_{q,k}\,-\,{I}_{q}^{\text{\ensuremath{\max}}^{(p)}}\leq0},\quad p=1,\ldots,P,\vspace{-.2cm}\label{eq:interference_constraints_multiple_PUs}
\end{equation}
}and\emph{\vspace{-.4cm}}

\emph{
\begin{equation}
I^{\,(p)}(\wh{\boldsymbol{{\tau}}},\,\mathbf{p},\,\wh{\boldsymbol{{\gamma}}})\triangleq\sum_{q=1}^{Q}{\displaystyle {\sum_{k=1}^{N}}\,\wh{P}_{q,k}^{\text{{miss}}}\left(\wh{\gamma}_{q,k},\wh{\tau}_{q}\right)\cdot w_{q,k}^{\,(p)}\cdot p_{q,k}\,-\,{I}^{\text{\ensuremath{\max}}^{(p)}}\leq0},\quad p=1,\ldots,P,\vspace{-.2cm}\label{eq:interference_constraints_multiple_PUs_global}
\end{equation}
}where ${I}_{q}^{\text{\ensuremath{\max}}^{(p)}}$ {[}or ${I}^{\text{\ensuremath{\max}}^{(p)}}$
{]} is the maximum average interference allowed to be generated by
the SU $q$ {[}or all the SU's{]} that is tolerable at the primary
receiver $p$; $w_{q,k}^{\,(p)}$'s are a given set of positive weights
(possibly different for each PU $p$); and $\wh{P}_{q,k}^{\text{{miss}}}\left(\wh{\gamma}_{q,k},\wh{\tau}_{q}\right)$
is the missed detection probability in the transformed variables,
defined in (\ref{eq:P_miss_new}). Average \emph{per-carrier} local/global
interference constraints can also be introduced \cite{Scutari-Pang_DSP11}. 

To cast the system design in the general formulation (\ref{eq:player q transformed 1})-(\ref{eq:price equilibrium_vector_form}),
we can proceed as in Sec. \ref{sub:Game-with-pricing}. Similarly
to (\ref{eq:set_Xq}) and (\ref{eq:map_interference}), we introduce
the feasible set of local constraints of each user $q$, denoted now
by $\wh{\mathcal{X}}_{q}$, 
\[
\wh{\mathcal{X}}_{q}\triangleq\left\{ (\wh{\tau}_{q},\,\mathbf{p}_{q},\,\wh{\boldsymbol{{\gamma}}}_{q})\in\mathcal{Y}_{q}\,\,|\,\, I_{q}^{\,(p)}(\wh{{\tau}_{q}},\,\mathbf{p}_{q},\,\wh{\boldsymbol{{\gamma}}}_{q})\leq0,\quad\forall p=1,\ldots,P\right\} ,
\]
with the convex part $\mathcal{Y}_{q}$ defined as in (\ref{eq:def_Y_q});
and the interference violation (column) \emph{vector} function $\mathbf{I}(\wh{\boldsymbol{{\tau}}},\,\mathbf{p},\,\wh{\boldsymbol{{\gamma}}})$
grouping all the global interference constraints (\ref{eq:interference_constraints_multiple_PUs_global}),
and defined as 
\begin{equation}
\mathbf{I}(\wh{\boldsymbol{{\tau}}},\,\mathbf{p},\,\wh{\boldsymbol{{\gamma}}})\triangleq\left(I^{\,(p)}(\wh{\boldsymbol{{\tau}}},\,\mathbf{p},\,\wh{\boldsymbol{{\gamma}}})\right)_{p=1}^{P}.\label{eq:vector interference function}
\end{equation}
Instead of having a single price variable, in the presence of multiple
global interference constraints, we associate a different price $\pi^{(p)}$
to each global interference constraint $I^{\,(p)}(\wh{\boldsymbol{{\tau}}},\,\mathbf{p},\,\wh{\boldsymbol{{\gamma}}})$.
Thus, we will have a vector of prices$-$the (column) vector $\boldsymbol{{\pi}}\triangleq(\pi_{p})_{p=1}^{P}-$to
be optimized along with multiple price clearance conditions (one for
each pair price/interference constraint). Using the above notation,
the resulting game theoretical formulation is the natural generalization
of (\ref{eq:player q transformed 1})-(\ref{eq:price equilibrium_vector_form})
and is given next. 

\hspace{-0.6cm}%
\framebox{\begin{minipage}[t]{0.98\columnwidth}%
\label{Players'-optimization-problems_multiple_PUs}\textbf{Players'
optimization} \textbf{problems}. The optimization problem of player
$q$ is:\vspace{-0.2cm}
\begin{equation}
\begin{array}{ll}
{\displaystyle {\operatornamewithlimits{\mbox{maximize}}_{\wh{\tau}_{q},\,\mathbf{p}_{q},\,\hat{{\boldsymbol{{\gamma}}}}_{q}}}} & \hat{R}_{q}\left(\wh{\tau}_{q},\,\mathbf{p},\,\hat{{\boldsymbol{{\gamma}}}}_{q}\right)-{\displaystyle {\displaystyle \boldsymbol{{\pi}}^{T}\,\mathbf{I}(\wh{\boldsymbol{{\tau}}},\,\mathbf{p},\,\wh{\boldsymbol{{\gamma}}})}}\vspace{-0.3cm}\\[0.25in]
\mbox{subject to} & \left(\wh{\tau}_{q},\,\mathbf{p}_{q},\,\hat{{\boldsymbol{{\gamma}}}}_{q}\right)\in\wh{\mathcal{X}}_{q}.\\[5pt]
\end{array}\vspace{-0.2cm}\label{eq:player q transformed_2-1}
\end{equation}
\textbf{Price equilibrium}. The price obeys the following complementarity
conditions:\vspace{-0.6cm}

\begin{flushleft}
\begin{equation}
0\,\leq\,\boldsymbol{{\pi}}\,\perp\,-\,\mathbf{I}(\wh{\boldsymbol{{\tau}}},\,\mathbf{p},\,\wh{\boldsymbol{{\gamma}}})\geq\mathbf{0}.\label{eq:price equilibrium_new-1}
\end{equation}

\par\end{flushleft}%
\end{minipage}}\bigskip

The analysis of this more general game can be addressed following
the same procedure as introduced in Sec. \ref{sec:NE-LNE-QE} for
the game (\ref{eq:player q transformed 1})-(\ref{eq:price equilibrium_vector_form}),
and thus is omitted.

\section{Numerical Results\label{sec:Numerical-Results}\vspace{-0.2cm}}

In this section, we provide some numerical results to illustrate our
theoretical findings. More specifically, we compare the performance
achievable at the QNE of the proposed game with those achievable by
the state-of-the-art decentralized \cite{Pang-Scutari-Palomar-Facchinei_SP_10}
(special cases are those in \cite{Luo-Pang_IWFA-Eurasip,Scutari-Palomar-Barbarossa_SP08_PI,Scutari-Palomar-Barbarossa_AIWFA_IT08})
and centralized \cite{SchmidtShiBerryHonigUtschick-SPMag} schemes
proposed in the literature for similar problems; such schemes \emph{do
not perform any sensing optimization} using thus all the frame length
for the transmission, and the QoS of the PUs is preserved by imposing
(deterministic) interference constraints (we properly modified the
algorithms in \cite{SchmidtShiBerryHonigUtschick-SPMag} to include
the interference constraints in the feasible set of the optimization
problem). Interestingly, the proposed design of CR systems based on
the\emph{ }joint optimization of the sensing and transmission strategies
is shown to outperform \emph{both centralized and decentralized current
CR designs}, which validates our new formulations. We then show an
example of the optimal sensing/throughput trade-off achievable at
the QNE. Finally, we compare the sum-throughput achievable by the
SUs in the presence of local and global interference constraints {[}game
formulation (\ref{eq:player q_individual_interference_constraints})
vs. (\ref{eq:penalized_game}){]}, which sheds light on the achievable
trade-off between signaling and performance. 

All the numerical results here are obtained using the algorithms described
in the companion paper \cite{Pang-Scutari-NNConvex_PII}. Note that
some of the proposed algorithms require some signaling among the users
in the form of consensus schemes; when this happens, some (\emph{finite})
time of the available (sensing-transmission) frame $T$, let us say
$T_{\text{{cons}}}$, needs to be allocated for this purpose. The
interesting result is that our consensus implementation is proved
to converge in a \emph{finite }number of iterations ($T_{\text{{cons}}}$
is then bounded). In order to make our numerical comparisons fair,
we included this time loss in the throughput functions, which are
still given by (\ref{eq:Opportunistic-throughtput}), but $1-\tau_{q}/T$
is replaced by $1-(\tau_{q}+T_{\text{{cons}}})/T$, with $T_{\text{{cons}}}\leq T$
and $\tau_{q}\leq T-T_{\text{{cons}}}$. \medskip{}

\noindent \textbf{Example \#1: How good is a QNE?} In Figure \ref{Fig_snr_gain}
we compare the performance achievable at a QNE of the game (\ref{eq:penalized_game})
with those achievable at the stationary solutions of the sum-rate
maximization problem subject to interference constraints \cite{SchmidtShiBerryHonigUtschick-SPMag}
{[}subplot (a){]} and at the NE of the game in \cite{Pang-Scutari-Palomar-Facchinei_SP_10}
{[}subplot (b){]}. We consider a hierarchical CR network where there
are two PUs (the base stations of two cells) and fifteen SUs, randomly
distributed in the cells. The (cross-)channels among the secondary
links and between the primary and the secondary links are simulated
as FIR filter of order $L=8$, where each tap has variance equal to
$1/L^{2}$; the available bandwidth is divided in $N=1024$ subchannels.
For the sake of simplicity we consider only individual interference
constraints (\ref{eq:individual_overal_interference_constraint}),
assuming the same interference limit $I^{\text{{max}}}$ for all SUs.
Moreover we impose for each player the same false-alarm rate over
all the subcarriers (possibly different among the players). To quantify
the throughput gain achievable at the QNE of the proposed game, in
Figure \ref{Fig_snr_gain}, we plot the (\%) ratio $(SR_{QE}-SR)/SR$
versus the (normalized) interference constraint bound $P/I^{\text{\ensuremath{\max}}}$
($P_{q}=P_{r}=P$ for all $q\neq r$), for different values of the
SNR detection $\texttt{snr}_{d}=\sigma_{I_{q,k}}^{2}/\sigma_{q,k}^{2}$,
where $SR_{QE}$ is the sum-throughput achievable at the QNE of the
proposed game (with $c=100$) whereas $SR$ is either the sum-rate
at a stationary solution of the social sum-rate maximization \cite{SchmidtShiBerryHonigUtschick-SPMag}
subject to interference constraints {[}subplot (a){]} or the sum-rate
at the NE of the game in \cite{Pang-Scutari-Palomar-Facchinei_SP_10}
{[}subplot (b){]}. From the picture, we clearly see that the proposed
joint optimization of the sensing and transmission parameters yields
a considerable performance improvement over the current state-of-the-art
CR \emph{centralized and decentralized} designs, especially when the
interference constraints are stringent. These results support  the
proposed novel formulation and concept of QNE.\medskip{}

\noindent \textbf{Example \#2: Sensing/throughput trade-off.} Figure
\ref{Fig_sens_through_tradeoff} shows an example of the expected
trade-off between the sensing time and the achievable throughput.
More specifically, in the picture, we plot the (normalized) sum-throughput
achieved at a QNE by a player of the game versus the (normalized)
\emph{common} sensing time, for different values of the (normalized)
total interference constraint (the setup is the same as in Figure
\ref{Fig_snr_gain}). According to the picture, there exists an optimal
duration for the (common) sensing time at which the throughput of
each SU is maximized. Moreover, as expected, more stringent interference
constraints impose lower missed detection probabilities as well as
false-alarm rates; requirement that is met by increasing the sensing
time (i.e., making the detection more accurate). This is clear in
the picture where one can see that the optimal sensing time duration
increases as the interference constraints increase. In the same figure
we also plot the sum-throughput achieved at the QNE of the game (\ref{eq:penalized_game})
setting $c=100$ (square markers in the plot). Interestingly, the
proposed approach based on a penalty function leads to performance
comparable with those achievable by a centralized approach that computes
the optimal common sensing time obtained by a discrete search of such
a time.\medskip{}

\noindent \textbf{Example \#3: Global vs. local interference constraints.
}Global interference constraints impose less stringent conditions
on the transmit power of the SUs than those imposed by the individual
interference constraints, implying better throughput performance of
the SUs (at the price however of more signaling among the SUs, as
quantified in our companion paper \cite{Pang-Scutari-NNConvex_PII}).
Figure \ref{Fig_exPSDprofile} confirms this intuition; in the subfigure(a)
we plot the average (normalized) sum-throughput of the SUs achievable
in the game (\ref{eq:player q_individual_interference_constraints})
{[}only local interference constraints{]} and (\ref{eq:penalized_game})
{[}global interference constraints{]} as a function of the maximum
tolerable interference at the primary receivers, within the same setup
of Figure 1 (the curves are averaged over 300 random i.i.d. Gaussian
channel realizations). In (\ref{eq:player q_individual_interference_constraints}),
the interference thresholds ${I}_{q}^{\text{{\ensuremath{\max}}}}$
are set ${I}_{q}^{\text{{\ensuremath{\max}}}}={I}^{\text{{\ensuremath{\max}}}}/Q$
for all $q$ and both PUs, so that all the SUs generate at the primary
receivers the same average interference level and the aggregate average
interference satisfies the imposed interference threshold ${I}^{\text{{\ensuremath{\max}}}}$.
In Figure \ref{Fig_exPSDprofile}(b), we show an example of the average
(normalized) interference profile measured at one of the (two) primary
receivers, obtained solving the game (\ref{eq:player q_individual_interference_constraints})
and (\ref{eq:penalized_game}). We clearly see from both pictures
that, as expected, global interference constraints are less conservative
than the local ones, yielding thus better performance of the SUs;
this however comes at the price of some signaling among the SUs. 
\begin{figure}
\center\includegraphics[height=6.5cm]{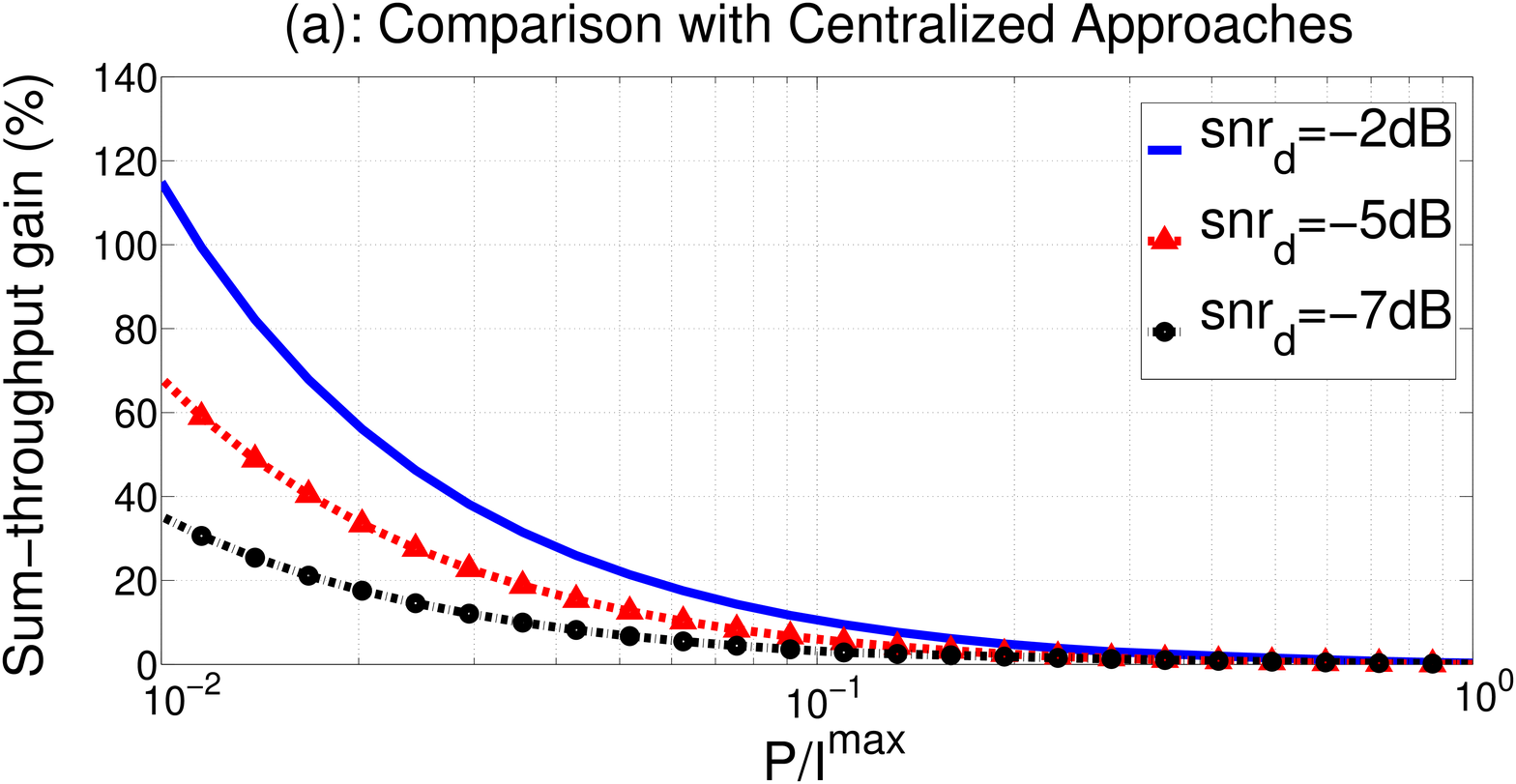}\vspace{0.5cm}\hspace{-0.5cm}\includegraphics[height=6.8cm]{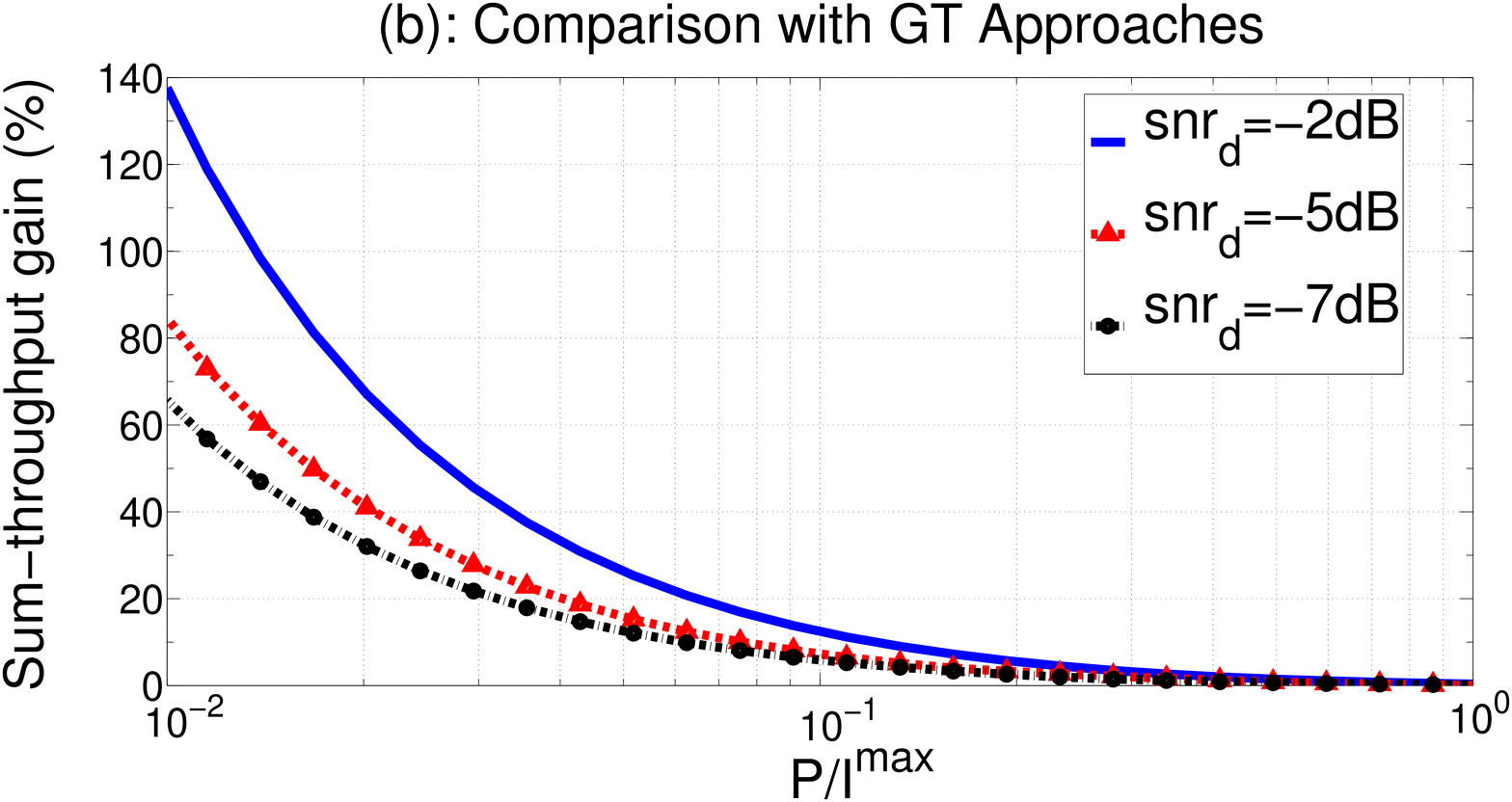}

{\footnotesize \caption{{\small Performance achievable at the QNE of (}\ref{eq:penalized_game}{\small ):
sum-rate gain (\%) versus the normalized interference, for different
values of the detection SNR and $c=100$}; {\small comparison with
centralized solutions \cite{SchmidtShiBerryHonigUtschick-SPMag} {[}subplot
(a){]} and game theoretical solutions \cite{Pang-Scutari-Palomar-Facchinei_SP_10}
{[}subplot (b){]}.}\label{Fig_snr_gain}}
}
\end{figure}
\begin{figure}[h]
\vspace{-0.2cm}\center\includegraphics[height=7.3cm]{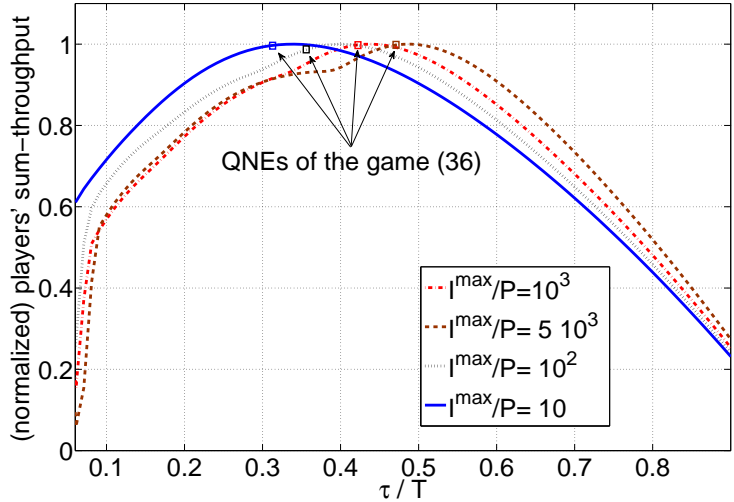}

\caption{{\small Example of optimal sensing/throughput trade-off: Normalized
}throughput{\small{} versus the normalized sensing time, for different
values of the (normalized) interference threshold. The square markers
correspond to the QNE of the game (}\ref{eq:penalized_game}{\small ),
achieved with $c=100$.}\label{Fig_sens_through_tradeoff}}
\end{figure}
\begin{figure}[H]
\vspace{-0.6cm}\hspace{-0.3cm}\includegraphics[height=7cm]{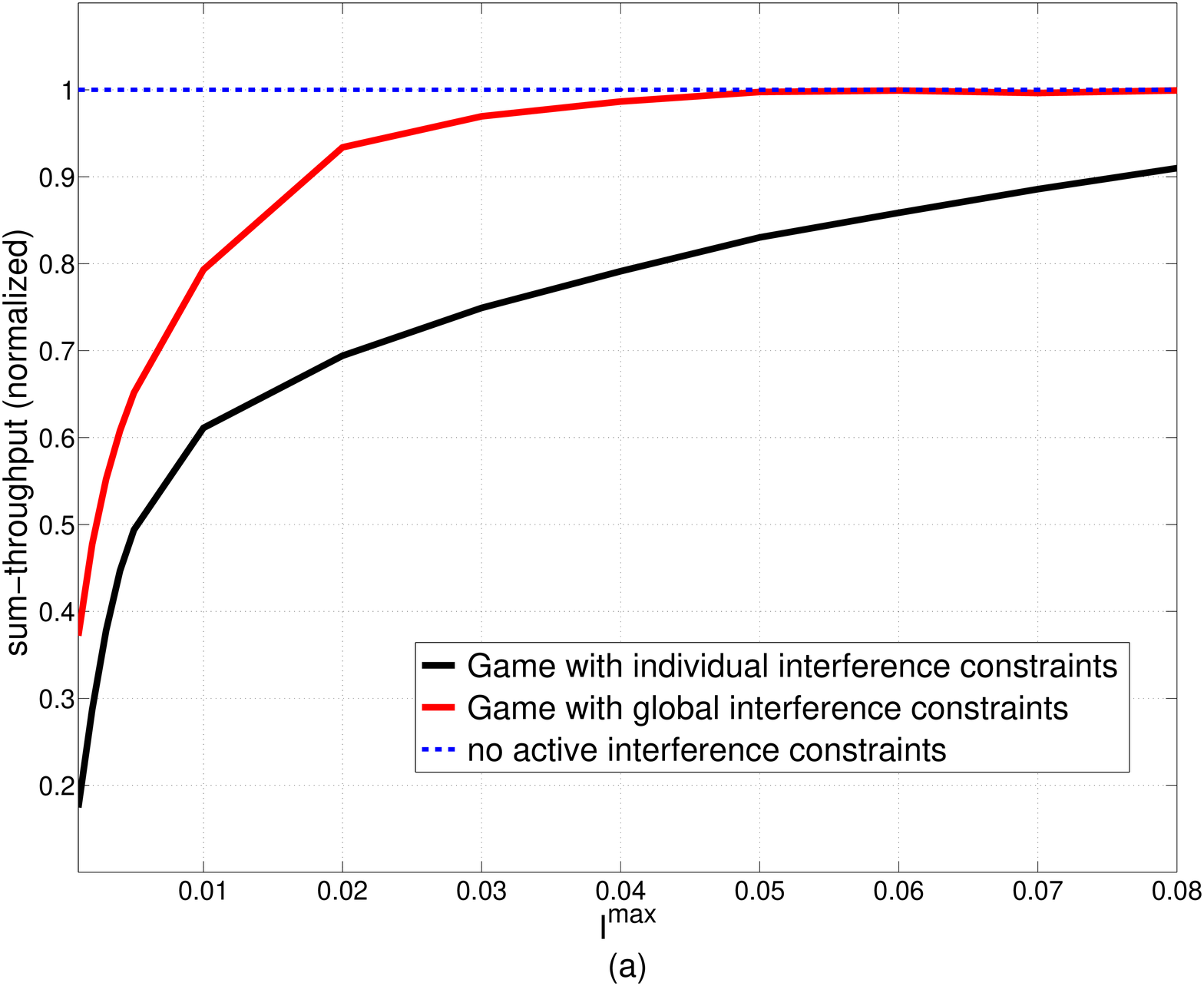}\hspace{-.8cm}\includegraphics[height=7cm]{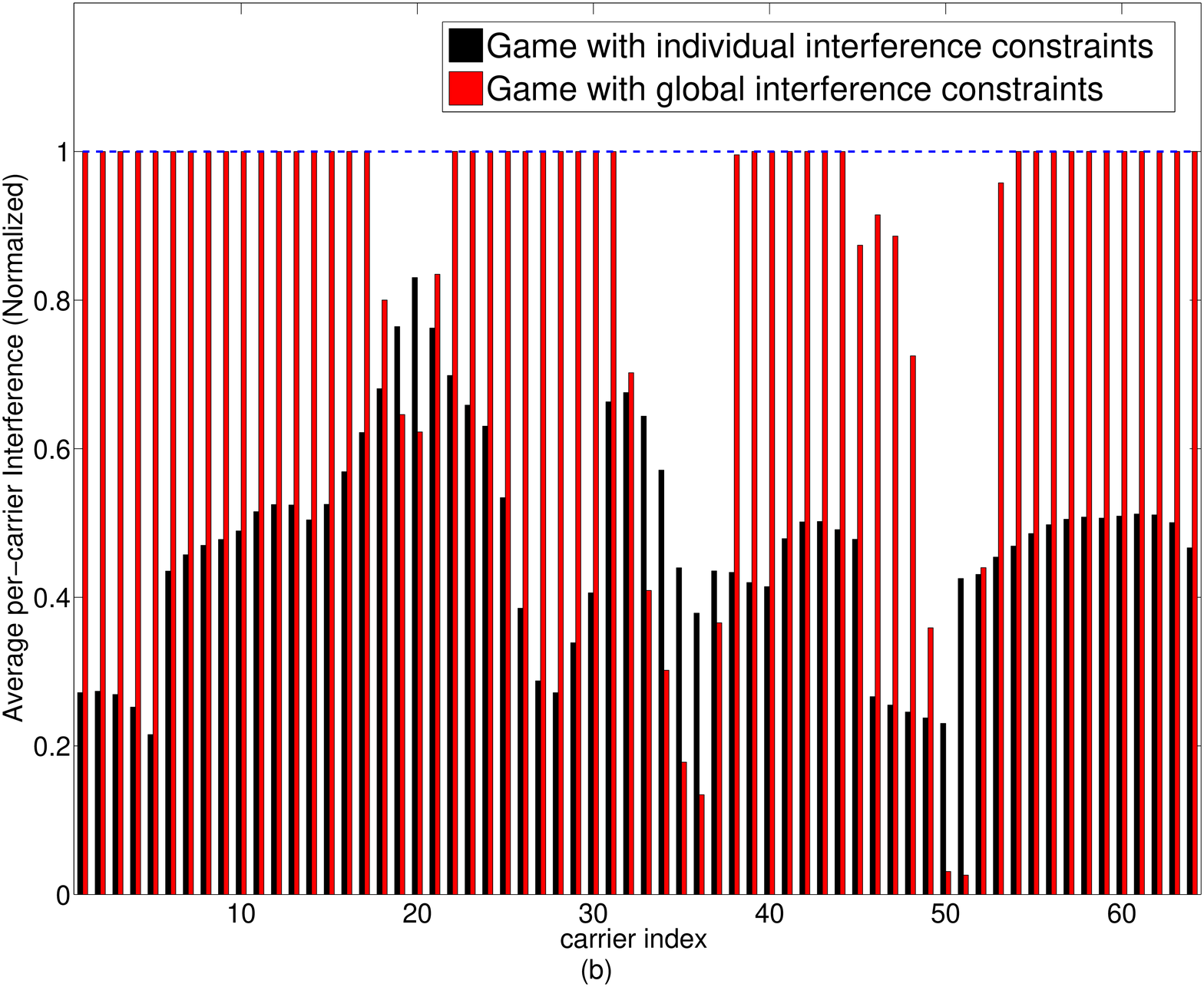}\vspace{-.2cm}

{\footnotesize \caption{{\small Global versus individual interference constraints {[}game
(\ref{eq:player q_individual_interference_constraints}) versus game
(\ref{eq:penalized_game}){]}. Subplot (a): average (normalized) sum-throughput
as a function of the maximum tolerable interference $I^{\max}$ at
the primary receivers. Subplot (b): Average per-carrier probabilistic
(normalized) interference generated by the solution of the games at
one of the PU's receiver. }\label{Fig_exPSDprofile}}
}
\end{figure}

\section{Conclusions\label{sec:Conclusions}}

\noindent In this paper, we proposed a novel class of Nash problems
wherein each SU aims to maximize his own opportunistic throughput
by choosing jointly the sensing duration, the detection thresholds,
and the vector power allocation over a multichannel link, under several
interference constraints, either local or global. In particular, to
enforce global interference constraints while keeping the optimization
as decentralized as possible, we proposed a pricing mechanism that
penalizes the SUs in violating the global interference constraints.
The resulting games belong to the class of nonconvex games with unbounded
pricing variables, whose analysis cannot be addressed using mathematical
tools from the existing game theory literature. To deal with these
issues, we proposed the use of a relaxed equilibrium concept$-$the
QNE$-$and studied its properties and connection with LNE. In particular,
we proved that the proposed games always have a QNE, even when a (L)NE
may not exist. We then validated the QNE both theoretically and numerically:
i) we proved that the QNE has some local optimality properties; and
ii) numerical results show the superiority of the proposed design
with respect to the state-of-the-art centralized and decentralized
resource allocation algorithms for CR systems. Distributed solution
schemes for computing such equilibria along with their convergence
properties are proposed and studied  in the companion paper \cite{Pang-Scutari-NNConvex_PII}. 

\appendix

\section*{Appendix: Proof of Theorem \ref{th:QE}\label{sec:Proof-of-TheoremQNE}}

\noindent Because of the space limitation, we provide only a sketch
of the proof. We need to show that the VI$(\mathcal{S},\boldsymbol{\Psi})$
in (\ref{eq:VI_ref}) has a solution. According to \cite[Proposition~2.2.3]{Facchinei-Pang_FVI03},
it is sufficient to find a tuple $(\mathbf{x}^{\text{{ref}}},\,\pi^{\text{{ref}}},\,\boldsymbol{{\lambda}}^{\text{{ref}}})\in\mathcal{S}$,
with $\mathbf{x}^{\text{{ref}}}\triangleq(\wh{\tau}_{q}^{\,\text{{ref}}},\,\mathbf{p}_{q}^{\text{\,{ref}}},\,\wh{\boldsymbol{{\gamma}}}_{q}^{\,\text{{ref}}})_{q=1}^{Q}$,
such that the set 
\begin{equation}
\mathcal{L}_{\leq}\,\triangleq\,\left\{ \,\mathbf{(\mathbf{x},\,\pi,\,\boldsymbol{{\lambda}})}\,\in\,\mathcal{S}\,\mid\,\left(\begin{array}{l}
\mathbf{x}-\mathbf{x}^{\,\text{{ref}}}\\
\pi-\pi^{\,\text{{ref}}}\\
\boldsymbol{{\lambda}}-\boldsymbol{{\lambda}}^{\,\text{{ref}}}
\end{array}\right)^{T}\boldsymbol{\Psi}\left(\mathbf{x}^{\,\text{{ref}}},\,\pi^{\,\text{{ref}}},\,\boldsymbol{{\lambda}}^{\,\text{{ref}}}\right)\,\leq\,0\,\right\} \label{eq:set_L}
\end{equation}
is bounded. We choose $(\mathbf{x}^{\text{\,{ref}}},\,\pi^{\,\text{{ref}}},\,\boldsymbol{{\lambda}}^{\,\text{{ref}}})$
by letting all its components to be zero except for the $\wh{\tau}$-
and $\wh{\gamma}$-components which we define as follows: $\wh{\gamma}_{q,k}^{\,\text{{ref}}}\triangleq\wh{\beta}_{q,k}$
and $\wh{\tau}_{q}^{\,{\rm ref}}\,\triangleq\wh{\tau}_{q}^{\,\max}$
for all $k=1,\ldots,N$ and $q=1,\ldots,Q.$ Invoking the feasibility
conditions (\ref{eq:nec_suff_feasibility_cond}) and using the above
definitions, we have: for all $k=1,\ldots,N$ and $q=1,\ldots,Q,$
\begin{equation}
\wh{\gamma}_{q,k}-\wh{\gamma}_{q,k}^{\,\text{{ref}}}\,\triangleq\,{\displaystyle {\frac{\sigma_{q,k|0}\,(\,\wh{\gamma}_{q,k}-\wh{\beta}_{q,k}\,)-(\,{\mu}_{{q,k}|1}-\mu_{{q,k}|0}\,)\,(\,\wh{\tau}_{q}-\wh{\tau}_{q}^{\,\text{{\rm ref}}}\,)}{{\sigma}_{q,k|1}}}\,\geq\,0.}\label{eq:zeta-ref}
\end{equation}
 Let $\mathbf{(\mathbf{x},\,\pi,\,\boldsymbol{{\lambda}})}\in\mathcal{L}_{\leq}$,
with the tuple $(\mathbf{x}^{\text{\,{ref}}},\,\pi^{\,\text{{ref}}},\,\boldsymbol{{\lambda}}^{\,\text{{ref}}})$
defined above. Clearly, the players' variables $\mathbf{x}=(\wh{\boldsymbol{{\tau}}},\,\mathbf{p},\,\wh{\boldsymbol{{\gamma}}})$
are bounded: for every $q=1,\ldots,Q$, we have that $(\wh{\boldsymbol{{\tau}}}_{q},\,\mathbf{p}_{q},\,\wh{\boldsymbol{{\gamma}}}_{q})\in\mathcal{X}_{q}$
implies 
\begin{equation}
\wh{\tau}_{q}^{\max}\leq\wh{\tau}_{q}\leq\wh{\tau}_{q}^{\max},\quad\,\mathbf{0}\leq\mathbf{p}\leq\mathbf{p}^{\max},\quad\mbox{and}\quad\wh{\beta}_{q,k}\leq\wh{\gamma}_{q,k}\leq{\displaystyle {\frac{\sigma_{q,k|1}\,\wh{\alpha}_{q,k}+(\,{\mu}_{{q,k}|1}-\mu_{{q,k}|0}\,)\,\wh{\tau}_{q}^{\max}}{{\sigma}_{q,k|0}}}\triangleq\wh{\gamma}_{q,k}^{\max},\,\,\forall k.}\label{eq:bounds}
\end{equation}
Next, we show that the multipliers $\boldsymbol{{\lambda}}$ and the
price $\pi$ are also bounded. Multiplying out $\left[(\mathbf{x}-\mathbf{x}^{\,\text{{ref}}})^{T},\right.$
$\left.(\pi-\pi^{\,\text{{ref}}})^{T},\,(\boldsymbol{{\lambda}}-\boldsymbol{{\lambda}}^{\,\text{{ref}}})^{T}\right]\,\boldsymbol{\Psi}\left(\mathbf{x}^{\,\text{{ref}}},\,\pi^{\,\text{{ref}}},\,\boldsymbol{{\lambda}}^{\,\text{{ref}}}\right)$,
using the expression of $\boldsymbol{\Psi}\left(\mathbf{x}^{\,\text{{ref}}},\,\pi^{\,\text{{ref}}},\,\boldsymbol{{\lambda}}^{\,\text{{ref}}}\right)$,
canceling some terms, and collecting the remaining terms, we deduce
{\small 
\begin{equation}
\begin{array}{rl}
0\geq & 2\,{\displaystyle {\sum_{q=1}^{Q}}}\dfrac{{\wh{\tau}_{q}\,(\,\wh{\tau}_{q}-\wh{\tau}_{q}^{\,{\rm ref}}\,)}}{f_{q}\, T_{q}}\,{\displaystyle {\sum_{k=1}^{N}}}\left(\,1-\mathcal{Q}(\wh{\gamma}_{q,k})\,\right)\, r_{q,k}\left(\mathbf{p}\right)\,+\,{\displaystyle {\sum_{q=1}^{Q}}}\,{\displaystyle {\sum_{k=1}^{N}}}\,\left(\,1-{\displaystyle {\frac{\wh{\tau}_{q}^{\,2}}{f_{q}\, T_{q}}}\,}\right)\,\left(\,\wh{\gamma}_{q,k}-\wh{\beta}_{q,k}\,\right)\,\mathcal{Q}^{\,\prime}(\wh{\gamma}_{q,k})\, r_{q,k}\left(\mathbf{p}\right)\,\vspace{-0.4cm}\\[0.4in]
 & -\,{\displaystyle {\sum_{q=1}^{Q}}}\,{\displaystyle {\sum_{k=1}^{N}}}\,\left(\,1-{\displaystyle {\frac{\wh{\tau}_{q}^{2}}{f_{q}\, T_{q}}}\,}\right)\,{\displaystyle {\frac{\left(\,1-\mathcal{Q}(\wh{\gamma}_{q,k})\,\right)\,|{H}_{qq}(k)|^{2}p_{q}(k)}{{\sigma}_{q,k}^{2}+\sum_{r=1}^{Q}\,|{H}_{rq}(k)|^{2}\, p_{r}(k)}}}-\,{\displaystyle {\sum_{q=1}^{Q}}}\,{\displaystyle {\sum_{k=1}^{N}}}\,\left(\,\lambda_{q}+{\pi}\right)\cdot w_{q,k}\cdot p_{q}(k)\,\mathcal{Q}^{\,\prime}(\wh{\gamma}_{q,k})\,\left[\,\wh{\gamma}_{q,k}-\wh{\gamma}_{q,k}^{\,\text{{ref}}}\,\right]\\[0.4in]
 & +{\displaystyle {\sum_{q=1}^{Q}}}\,\lambda_{q}\,{I}_{q}^{\,\max}+{\pi}\cdot{I}^{\,\text{{max}}}\vspace{-1.3cm}.\\[0.3in]
\end{array}\label{eq:KKT_existence_QNE}
\end{equation}
} 

\noindent Using the fact that the third and the fourth term on the
right hand side of (\ref{eq:KKT_existence_QNE}) are nonpositive {[}the
nonpositivity of the fourth term comes from (\ref{eq:zeta-ref}) and
$\mathcal{Q}^{\,\prime}(\wh{\gamma}_{q,k})\leq0${]}, we have

\noindent {\small 
\[
\begin{array}{l}
{\displaystyle {\sum_{q=1}^{Q}}}\,\lambda_{q}\,{I}_{q}^{\,\max}+{\pi}\cdot{I}^{\,\text{\ensuremath{\max}}}\,\leq\,{\displaystyle {\displaystyle {\sum_{q=1}^{Q}}}{\displaystyle {\sum_{k=1}^{N}}}}\left[\left(\dfrac{2\,{\wh{\tau}_{q}\,(\,\wh{\tau}_{q}^{\,{\rm ref}}-\wh{\tau}_{q}\,)}}{f_{q}\, T_{q}}\right)\left(\,1-\mathcal{Q}(\wh{\gamma}_{q,k})\,\right)+\left(\,1-{\displaystyle {\frac{\wh{\tau}_{q}^{\,2}}{f_{q}\, T_{q}}}\,}\right)\left(\,\wh{\gamma}_{q,k}-\wh{\beta}_{q,k}\,\right)\left|\mathcal{Q}^{\,\prime}(\wh{\gamma}_{q,k})\right|\right]r_{q,k}\left(\mathbf{p}\right)\\[0.3in]
\leq{\displaystyle {\sum_{q=1}^{Q}}}\,{\displaystyle {\sum_{k=1}^{N}}}\left[\,{\displaystyle {\frac{2\,\wh{\tau}_{q}^{\,\max}\,(\,\wh{\tau}_{q}^{\,\max}-\wh{\tau}_{q}^{\,\min}\,)}{f_{q}\, T_{q}}}+\left(\,1-\left({\displaystyle {\frac{\wh{\tau}_{q}^{\,\min}}{\sqrt{f_{q}\, T_{q}}}}}\right)^{2}\,\right)\,{\displaystyle {\frac{\wh{\gamma}_{q,k}^{\max}-\wh{\beta}_{q,k}}{\sqrt{2\,\pi}}}\,}}\right]\,\log\left(1+{\displaystyle {\frac{|{H}_{qq}(k)|^{2}p_{q,k}^{\,\max}}{{\sigma}_{q,k}^{2}}}}\right)
\end{array}
\]
}where the last inequality follows from the bounds (\ref{eq:bounds})
and $\left|\mathcal{Q}^{'}(\wh{\gamma}_{q,k})\right|\leq1/\sqrt{{2\pi}}$,
which establishes the boundedness of the elements in the set $\mathcal{L}_{\leq}$.

\noindent \indent The last part of the theorem follows from Proposition
\ref{pr:LNE implies QE} and the easy check that a trivial QNE cannot
satisfy (\ref{eq:VI_ref}).\hfill{}$\Box$

\begin{spacing}{1.1000000000000001}
\noindent \textcolor{black}{\footnotesize \bibliographystyle{IEEEtran}
\bibliography{scutari_refs}
}{\footnotesize \par}
\end{spacing}

\end{document}